\documentclass[%
twocolumn,
superscriptaddress,
showpacs,
showkeys,
nofootinbib,
nobibnotes,
 amsmath,
 amssymb,
 aps,
pra,
]{revtex4-1}

\usepackage{algorithm}
\usepackage{algorithmic}
\usepackage{graphicx}
\usepackage{dcolumn}
\usepackage{bm}
\usepackage{amsthm} 
\usepackage{hyperref}

\newcommand{\be}{\begin{equation}}
\newcommand{\ee}{\end{equation}}
\newcommand{\bea}{\begin{eqnarray}}
\newcommand{\eea}{\end{eqnarray}}

\newtheorem{thm}{Theorem}

\newtheorem{observation}{Observation}

\newtheorem{conjecture}{Conjecture}
\newtheorem{Problem}{Problem}

\input{Qcircuit}

\begin{document}

\title{
   Efficient Decomposition of Single-Qubit Gates into $V$ Basis Circuits
}

\author{Alex Bocharov}
  \email{alexeib@microsoft.com}
  \affiliation{Quantum Architectures and Computation Group,
 Microsoft Research, Redmond, WA 98052 USA}

  \author{Yuri Gurevich}
  \email{gurevich@microsoft.com}
  \affiliation{Research In Software Engineering Group,
 Microsoft Research, Redmond, WA 98052 USA}

  \author{Krysta M. Svore}
   \email{ksvore@microsoft.com}
   \affiliation{Quantum Architectures and Computation Group,
 Microsoft Research, Redmond, WA 98052 USA}

\begin{abstract}
We develop the first constructive algorithms for compiling single-qubit unitary gates into circuits over the universal $V$ basis.
The $V$ basis is an alternative universal basis to the more commonly studied $\{H,T\}$ basis.
We propose two classical algorithms for quantum circuit compilation: the first algorithm has expected polynomial time (in precision $\log(1/\epsilon)$) and offers a depth/precision guarantee that improves upon state-of-the-art methods for compiling into the $\{H,T\}$ basis by factors ranging from $1.86$ to $\log_2(5)$.
The second algorithm is analogous to direct search and yields circuits a factor of 3 to 4 times shorter than our first algorithm, and requires time exponential in $\log(1/\epsilon)$;
however, we show that in practice the runtime is reasonable for an important range of target precisions.
\end{abstract}


\maketitle

\section{Introduction}
Determining the optimal fault-tolerant compilation, or decomposition, of a quantum gate is critical for designing a quantum computer.
Decomposition of single-qubit unitary gates into the $\{H, T\}$ basis has been well studied in recent years.
However, there have been few studies of decomposing into alternative bases, which may offer significant improvements in circuit depth or resource cost.
In this work, we consider the task of decomposing a single-qubit unitary gate into a sequence of gates drawn from the \emph{V basis}, first introduced in Refs.~\onlinecite{LPSI,LPSII}.
Historically, this basis was the first shown to be {\it efficiently universal}, in that the length of the decomposition sequence is guaranteed to be of depth $O(\log(1/\epsilon))$ \cite{HRC}, however the proof did not offer a constructive algorithm.
Recently, it has been shown that $\{H, T\}$ is also efficiently universal \cite{Selinger},\cite{KMM1231}, and the proofs are constructive.
In this work, we show that despite recent advances for the $\{H,T\}$ basis, the $V$ basis allows for significantly shorter decomposition circuits.

We present two algorithms for compilation into the $V$ basis.
The first algorithm approximates single-qubit unitaries over the set consisting of the $V$ basis and the Clifford group;
the second approximates over the set consisting of the $V$ basis and the Pauli gates.
The first algorithm runs in expected polynomial time and delivers $\epsilon$-approximations with circuit depth $\leq 12 \, \log_5(2/\epsilon)$.
%
The second algorithm produces $\epsilon$-approximations with circuit depth
$\leq 3\log_{5}{\left(1/\epsilon \right)}$ for most single-qubit unitaries,
and approximations of circuit depth $4\log_{5}{\left(2/\epsilon \right)}$ for edge cases.
The compilation time is linear in $1/\epsilon $
and thus exponential in $\log(1/\epsilon)$, however, in practice we find extremely short circuits (of length $L=28$) at precision level $\epsilon=3\ast {10}^{-7}$ with merely 1 minute of classical CPU time and modest space usage.

This work presents yet another alternative to Solovay-Kitaev decomposition, and produces circuits with lengths matching the proven lower bound of $\Omega(\log(1/\epsilon))$ \cite{HRC}.
%
We note that our motivation for studying decomposition into the $V$ basis stems from two sources.
The first was the proof in Ref.~\onlinecite{HRC} that it is efficiently universal.
The second was the recent protocol for distillation of non-stabilizer states \cite{DCS}, which gives the first known fault-tolerant implementation of one of the $V$ basis gates using only magic states, Clifford operations, and measurements.


\section{Related Work}

Recently, dramatic improvements have been achieved in quantum circuit compilation, in particular in the area of single-qubit decomposition.
We highlight four developments that are particularly relevant for interpreting our work in a more general context.

The Programmable Ancilla Rotation (PAR) method for implementing arbitrary single-qubit rotations by resource state teleportation \cite{CJetAl} underlines the tradeoff of performing an approximating circuit directly on the target \emph{logical} qubit versus on resource ancilla states followed by a teleportation protocol to interact with the target qubit.
An advantage of the method is that ancilla factories can be employed which prepare resource states for later use, in exchange for performing a probabilistic circuit on the target qubit which may require several attempts prior to success.
The actual cost of approximating a single-qubit unitary with this method is measured in terms of the number of resource states and the number of attempts required for success.

More recently, a technique for distilling non-stabilizer states was introduced in Ref.~\onlinecite{DCS} and shown to enable approximation of any single-qubit unitary.  This protocol also uses state teleportation and can achieve on average constant circuit depth.  A key consequence of this work is the ability to prepare a state that enables the fault-tolerant implementation of the $V$ basis gates.

Recent research on the characterization of $\langle H,T \rangle$ circuits \cite{KMM12,KMMb12} has lead to a seminal decomposition result:
a constructive algorithm for efficient ancilla-free compilation of a given single-qubit unitary into the $\langle H,T \rangle$ basis, with a corresponding $T$-count guarantee of the form $4\log_2(1/\epsilon) + 11$ for $Z$ rotations and $12\log_2(1/\epsilon) + K$, where $K \sim 33$ for general unitaries \cite{Selinger}.
Further improvements to this algorithm were shown in Ref.~\onlinecite{KMM1231}, which presents a less efficient compilation method that produces shorter ancilla-free approximation circuits with an expected $T$ count of $9.63\log_2(1/\epsilon)-20.79$.

Our direct search algorithm (Section \ref{sec:directalg}) produces $\epsilon$-approximation circuits with a $V$ count of $3\log_5(1/\epsilon)$ in most cases and $4\log_5(2/\epsilon)$ in edge cases.
If a fault-tolerant $V$ gate has the same cost as a fault-tolerant $T$ gate, then this algorithm gives state-of-the-art circuit depth asymptotics.
Figure \ref{fig:RelWork} plots the $T$ count ($V$ count)\footnote{For illustrative purposes, here we assume one $T$ gate has the same cost as one $V$ gate.} of the approximation circuits versus the precision $\epsilon$ for several state-of-the-art $\{H,T\}$-based methods and the $V$-based algorithms presented in Sections \ref{sec:polytime} and \ref{sec:directalg}.
The solid blue curve plots the theoretical bound for the algorithm given in Ref.~\onlinecite{Selinger}.
The dashed red curve is based on interpolation of the experimental results given in Ref.~\onlinecite{KMM1231}.
The dashed green curve plots the theoretical bound (matched by experimental data) for decomposition into the $V$ basis using our randomized algorithm (Section \ref{sec:polytime}).
The double black curve plots the average experimental results over 1000 random unitaries from decomposition into the $V$ basis using our direct search algorithm (Section \ref{sec:directalg}).

From this plot, we see that the $T$ count is substantially lower for a given precision when compiling into the $V$ basis.
These curves serve as evidence of the potential improvements in circuit decomposition by considering other bases,
and hopefully motivates research in determining an optimal and low-cost fault-tolerant implementation of a $V$ gate.
To the best of our knowledge, there is not yet a fault-tolerant implementation of the $V$ gate that has cost equal to that of the $T$ gate.

One possible exact implementation \cite{DCS} requires on average a constant depth of 3 per $V$ gate, but in turn requires an ``offline" cost in $T$ gates and is probabilistic.
If the protocol succeeds (which only occurs half of the time), then the cost per $V$ gate is only $5.35$ $T$ gates, making the algorithm competitive with \cite{Selinger} and \cite{KMM1231}.
However, if the protocol fails, then the cost increases.
Details on this implementation of a $V$ gate are given in Appendix \ref{app:Vgate}.
In order for decomposition into the $V$ basis to be competitive with state-of-the-art $\langle H,T\rangle$ decomposition,
it is necessary to determine an exact, fault-tolerant $V$ gate implementation with a cost less than the cost of 6 $T$ gates.
We proceed by describing the two algorithms for compiling into the $V$ basis.

\begin{figure}[tb]
  \centering
  \includegraphics[width=3.5in]{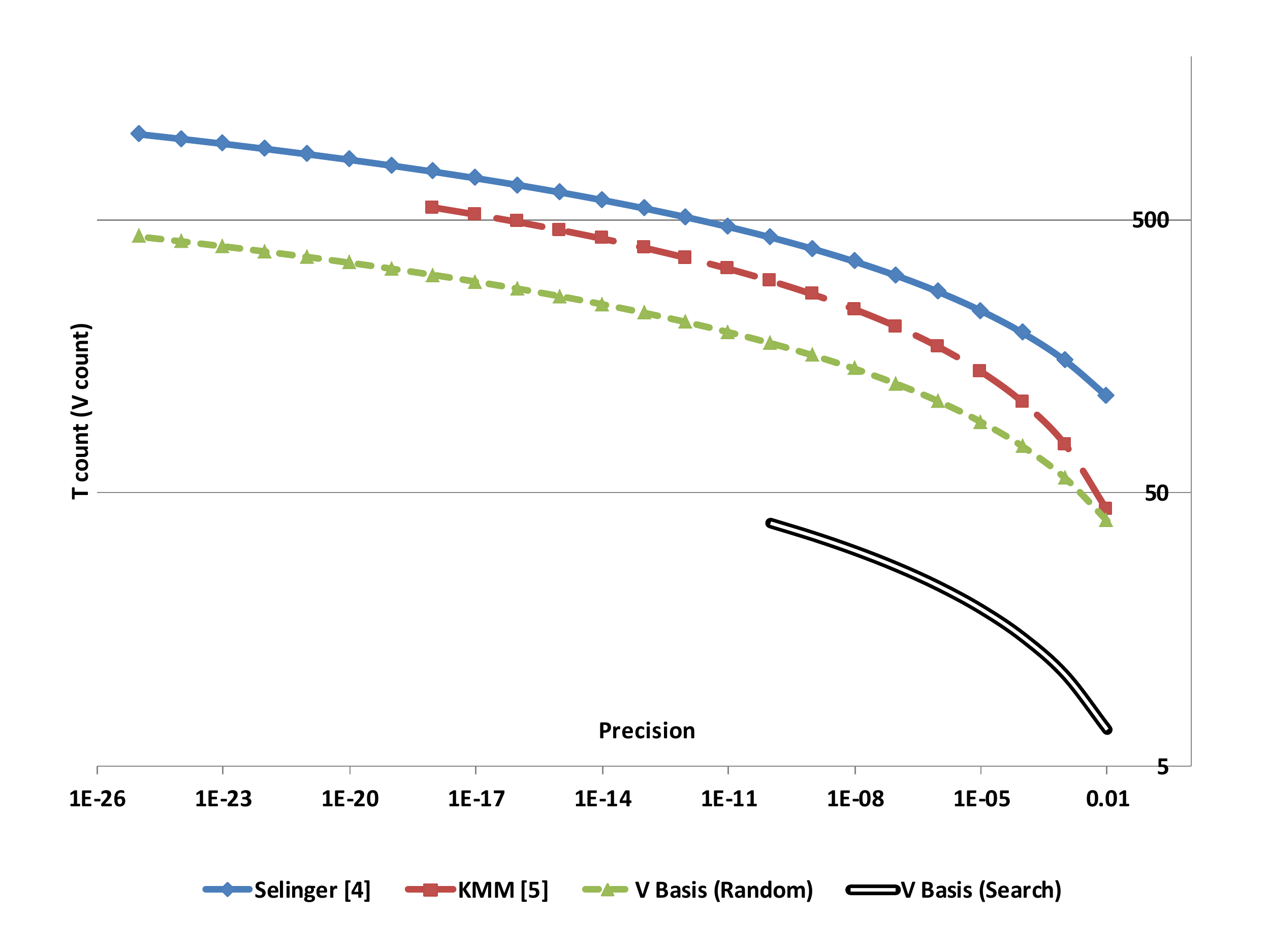}
  \caption[Depth/precision chart for four methods]{$T$ count ($V$ count) versus precision $\epsilon$ for state-of-the-art single-qubit decomposition methods.
  Algorithm in Ref.~\onlinecite{Selinger} (solid blue curve), algorithm in Ref.~\onlinecite{KMM1231} (dashed red curve),
  decomposition into the $V$ basis using the randomized algorithm (Section \ref{sec:polytime}; dashed red curve),
  decomposition into the $V$ basis using the direct search algorithm (Section \ref{sec:directalg}; dashed green curve).
  A $T$ gate is assumed to have equal cost to a $V$ gate.
}
\label{fig:RelWork}
\end{figure}



\section{Definitions and Key Theorems}
The efficiently universal single-qubit unitary basis introduced in Refs.~\onlinecite{LPSI,LPSII}
and further developed in Ref.~\onlinecite{HRC} consists of the following six special unitaries:
\begin{eqnarray*}
V_{1}=(I+2\, i\, X)/\sqrt 5 &,&\ V_{1}^{-1}=(I-2\, i\, X)/\sqrt 5,\\
V_{2}=(I+2\, i\, Y)/\sqrt 5 &,&\ V_{2}^{-1}=(I-2\, i\, Y)/\sqrt 5,\\
V_{3}=(I+2\, i\, Z)/\sqrt 5 &,&\ V_{3}^{-1}=(I-2\, i\, Z)/\sqrt 5.
\end{eqnarray*}
We call this basis the {\em $V$ basis}.

The subgroup $\langle V \rangle \subset SU(2)$ generated by this basis is everywhere dense in $SU(2)$ and thus $\{ V_i , V_i^{-1}, i=1,2,3\}$ is a universal basis.

Let the set of \emph{W circuits} be the set of those circuits generated by this basis and the Pauli
matrices $I, X, Y, Z$.

It is important to note that the monoid $\langle W \rangle=\langle X,Y,Z,V_1,V_2,V_3 \rangle \subset SU(2)$ contains all of the $\{ V_i^{-1}, i=1,2,3\}$ and thus is in fact a subgroup of $SU(2)$ containing $\langle V \rangle$.

$W$ circuits constitute a slight liberalization of the approach in Ref.~\onlinecite{HRC}, where only circuits in the $V$ basis are considered.
Our justification for the liberalization is that the Pauli operators are a staple of any quantum computing architecture and can be implemented fault tolerantly at a very low resource cost in comparison to a non-Clifford group gate.

It is also noted that the single-qubit Clifford group $\mathcal{C}$ in combination with {\it any} of the six $V$ matrices generates a monoid $\langle \mathcal{C}+V \rangle \subset SU(2)$ that is in fact a group, containing $\langle W \rangle$.

We call the number of $V$ gates the
\emph{$V$ count} of a circuit and denote it as $V_c$.
It is easy to show that an irreducible $W$ circuit contains at most one
non-identity Pauli gate. Thus, if $v$ is the $V$ count of such a circuit, then the
overall depth of the circuit is either $v$ or $v+1$.

Throughout, we use \emph{trace distance} to measure the distance between two unitaries $U,V \in PSU(2)$:
%
\be
dist(U,V) = \sqrt{1-|tr(U V^{\dagger}|/2},
\ee
%
and call the distance between a target unitary and the approximating unitary the \emph{precision} $\epsilon$.

According to \cite{HRC}, any single-qubit unitary can be approximated to a
given precision $\epsilon$ by a $V$ circuit of depth $O(\log \left( \frac{1}{\epsilon
} \right))$,
however the proof in Ref.~\onlinecite{HRC} is non-constructive, and no algorithm for
actual synthesis of the approximating circuits has yet been shown.
%
%
Here we develop effective solutions for synthesizing $W$-circuit
approximations of single-qubit unitaries.

Our solutions are based on the following theorem:
\begin{thm}\label{Proposition 1}
A single-qubit unitary gate $U$ can be exactly represented as a $W$ circuit of
$V$ count $V_c\le L\,$ if and only if it has the form
\be
U=\left( a\, I+b\, i\,
X+c\, i\, Y+d\, i\, Z \right)5^{-L/2},
\ee
where $a, b, c, d$ are integers such that $a^{2}+b^{2}+c^{2}+d^{2}=5^{L}$.
\end{thm}

Thm \ref{Proposition 1} follows from Thm \ref{thm2} given below, which also gives rise to a simple
constructive procedure for synthesizing a $W$ circuit that represents such
a $U$. 

We begin by sketching a linear-time subalgorithm for exact $W$-circuit synthesis that employs
arithmetic of Lipschitz quaternions \cite{Lip,CS}.
More specifically, consider the group $W$ of quaternions generated by
\begin{equation}
\label{eq1}
\pm 1,\pm \mathbf{i},\pm \mathbf{j},\pm \mathbf{k},1\pm 2\mathbf{i},\, 1\pm
2\mathbf{j},1\pm 2\mathbf{k}.
\end{equation}

Then the following holds:
\begin{thm}\label{thm2}
(1) $W$ is equal to the set of Lipschitz quaternions with norms
 $5^{l}, ( l\in \mathbb Z,\, l\ge 0)$.
(2) Consider the group $W_1 = \{w/\sqrt{norm(w)} | w \in W \}$. Then the subgroup of gates in $PSU(2)$ representable as exact  $W$-circuits is isomorphic to the central quotient $W_1/Z(W_1)$ where $Z(W_1)={\mathbb Z}_2=\{1,-1\}$.
\end{thm}

\begin{proof}

(1) We recall that the quaternion norm is multiplicative and that $\pm 1,\pm \mathbf{i},\pm \mathbf{j},\pm \mathbf{k}$ are the only
Lipschitz quaternions of norm 1. Thus statement (1) is true for $l=0$.

We prove it for $l=1$: More specifically, let  $q = a + b \, \mathbf{i} + c \, \mathbf{j} + d \, \mathbf{k}, a,b,c,d \in \mathbb Z$ and
$norm(q)= a^2+b^2+c^2+d^2=5$.

Decompositions of $5$ into sums of squares of four integers are easily enumerated and we conclude that exactly two of the coefficients in the list $\{a,b,c,d\}$ are zero, exactly one is $\pm 1$, and exactly one is $\pm 2$.

If $a= \pm 1$ then we observe that $q$ is equal to one of $1\pm 2\mathbf{i},\, 1\pm
2\mathbf{j},1\pm 2\mathbf{k}, -(1\pm 2\mathbf{i}),\, -(1\pm
2\mathbf{j}),-(1\pm 2\mathbf{k})$ and thus belongs to $W$.

If one of $b,c,d$ is $\pm 1$ we reduce the proof to the previous observation by multiplying $q$ times one of $\mathbf{i}, \mathbf{j}, \mathbf{k}$.

For example, if $c = \pm 1$ , then the real part of $-\mathbf{j} q$ is equal to $c = \pm 1$.

Consider now a quaternion $q$ with $norm(q) = 5^l , l \geq 1$.

Let $q=p_1 \, ...\, p_m$ be a prime quaternion factorization of $q$.
Since $5^l = norm(q) = norm(p_1) \, ...\, norm(p_m)$ , for each $i = 1,...m$ the $norm(p_i)$ is either 5 or 1.
As we have shown above (considering $l=0,1$), in either case $p_i \in W$.

(2) Effective homomorphism $h$ of $W_1$ onto the W-circuits is the multiplicative completion of the following map:
\begin{eqnarray*}
\mathbf{i} &\rightarrow& i \, X\\
\mathbf{j} &\rightarrow& i \, Y\\
\mathbf{k} &\rightarrow& i \, Z\\
(1\pm 2\mathbf{i})/\sqrt{5} &\rightarrow& (1 \pm 2 \, i \, X)/\sqrt{5}\\
(1\pm 2\mathbf{j})/\sqrt{5} &\rightarrow& (1 \pm 2 \, i \, Y)/\sqrt{5}\\
(1\pm 2\mathbf{k})/\sqrt{5} &\rightarrow& (1 \pm 2 \, i \, Z)/\sqrt{5}.
\end{eqnarray*}
The correctness of this definition of homomorphism $h$ is verified by direct comparison of multiplicative relations between the generators of
$W_1$ and $ g(W) = \{i \, X, i \, Y, i \, Z, (1 \pm 2 \, i \, X)/\sqrt{5}, (1 \pm 2 \, i \, Y)/\sqrt{5}, (1 \pm 2 \, i \, Z)/\sqrt{5} \}$.
These relations happen to be identical.

$h$ is an epimorphism since all of the generators $g(W)$ of the W-circuits group are by design in its image.

The characterization of $\mbox{\em Ker}(h)$ is derived from representation of quaternions as orthogonal rotations of the 3-dimensional Euclidean space.

Let us arbitrarily map the units $\mathbf{i},\mathbf{j},\mathbf{k}$ into vectors of an orthonormal  basis in the Euclidean space and let us label the corresponding basis vectors $e(\mathbf{i}),e(\mathbf{j}),e(\mathbf{k})$.
For a quaternion with zero real part $p=b \, \mathbf{i} + c \, \mathbf{j} + d \, \mathbf{k}$ we write $e(p) = b *e(\mathbf{i}) + c *e(\mathbf{j}) + d *e(\mathbf{k})$.

Let $H_1$ be the group of quaternions of norm 1 and $g:H_1 \rightarrow SO(3)$ be the representation defined as
$g(q)[e(b)] =  e(q * b * q^{-1})$.

It is known \cite{CS} that $g(q)$ is an orthogonal rotation; $g$ is a representation of the group of quaternions of norm 1 and that the kernel of this representation is the cyclic group ${\mathbb Z}_2 = \{1,-1\}$.

The group of quantum gates $PSU(2)$ also has a standard orthogonal representation stemming from its adjoint representation on the Lie algebra $\mathfrak{psu}(2) = \mathfrak{su}(2)= \mathfrak{so}(3)$.
More specifically if $\mathfrak{psu}(2)$ is regarded as the algebra of zero-trace Hermitian matrices then
$ad: PSU(2) \rightarrow Aut(\mathfrak{psu}(2))$, where $Aut$ is the automorphism, is $ad(u)[m] = u \, m \, u^{-1}$.

The adjoint representation of $PSU(2)$ is faithful.

If we regard the above homomorphism $h$ as the homomorphism $h: W_1 \rightarrow PSU(2)$ then it is immediate that $ad \, h = g$ on $W_1$.
Since $ad$ is faithful, i.e., injective,  the kernel of $h$ coincides with $\mbox{\em Ker}(g) = Z(W_1) = (\mathbb {Z})_2 = \{-1,1\}$.
\end{proof}

Lipschitz quaternions form a division ring, and in view of Thm \ref{thm2}, a quaternion with norm equal to $5^{l}$ can be decomposed into a product of generators in Eq \ref{eq1} in $l$ trial division steps.

The \emph{decomposition subalgorithm} (Algorithm \ref{alg:sub}) is thus as follows, with input being a Lipschitz quaternion $q$ of norm $5^{l}$:
%
\begin{algorithm}[H]
\caption{Decomposition Subalgorithm}
\label{alg:sub}
\algsetup{indent=2em}
\begin{algorithmic}[1]
\REQUIRE{A quaternion $q$ with norm $5^l$}
\STATE{$ret \leftarrow$ empty list}
\WHILE{$norm(q) > 0$}
\STATE{find $d$ in $\{1\pm2\mathbf{i},1\pm2\mathbf{j},1\pm 2\mathbf{k}\}$ such that $d$ divides $q$}
\STATE{$ret \leftarrow \{d\} + ret$}
\STATE{$q \leftarrow q/d$} //divides $norm(q)$ by $5$
\ENDWHILE
\IF{$q \neq 1$}
\STATE{$ret \leftarrow q + ret$}
\ENDIF
\RETURN{$ret$}
\end{algorithmic}
\end{algorithm}

Now, given a unitary $U$ as described in Thm \ref{Proposition 1}, we associate with it
the quaternion
$q=a+b\, \mathbf{i}+c\, \mathbf{j}+d\, \mathbf{k}$ that has norm
$5^{L}$ and thus belongs to the subgroup $W$.
It is easy to translate the factorization of $q$ in the basis given in Eq \ref{eq1} into a
factorization of $U$ in the $W$ basis.
Thus, the approximation of a target unitary gate $G$ by a $W$ circuit is
constructively reduced to approximating $G$ with a unitary $U$ as described in
Thm \ref{Proposition 1}.

\section{Randomized Approximation Algorithm}
\label{sec:polytime}
In this section, we present an algorithm for decomposing single-qubit unitaries into a circuit in the set $\langle \mathcal{C} + V \rangle$,
where $\mathcal{C}$ is the set of single-qubit Clifford gates and $V$ is one of the $V$ gates.
The expected polynomial runtime is based on a conjecture, for which we have developed ample empirical evidence (based on computer simulation).
We first present the conjecture and relevant number theory background, and then present the compilation algorithm.

\subsection{Number Theory Background}
Let $N$ be a large positive integer, and $\Delta$ be a relatively small fixed offset value.
Let $x,y$ be standard coordinates on a 2-dimensional Euclidean plane.

We introduce the circumference
\begin{equation*}
C(N,\Delta)= \{ (x,y)\ |\ x^2 + y^2 = (\sqrt{N}-\Delta)^2 \}.
\end{equation*}
Let $R(N,\Delta)$ be the circular ring of width $\Delta$ defined as
\begin{equation*}
R(N,\Delta)=\{(x,y)\ |\ (\sqrt{N}-\Delta)^2 < x^2 + y^2 < N \}.
\end{equation*}

Consider a tangent straight line at any point on the circumference $C(N,\Delta)$. The line divides the plane into two half-planes and let $P_{+}$ be the half-plane that does not contain the origin.

Next, define the circular segment
\begin{equation*}
A(N,\Delta, P_{+}) = R(N,\Delta) \bigcap P_{+}.
\end{equation*}

The ring $R(N,\Delta)$ and the circular segment $A(N,\Delta, P_{+})$ are shown schematically in Figure \ref{fig:RingAndSector}.
\begin{figure}[tb]
  \centering
  \includegraphics[width=3.5in]{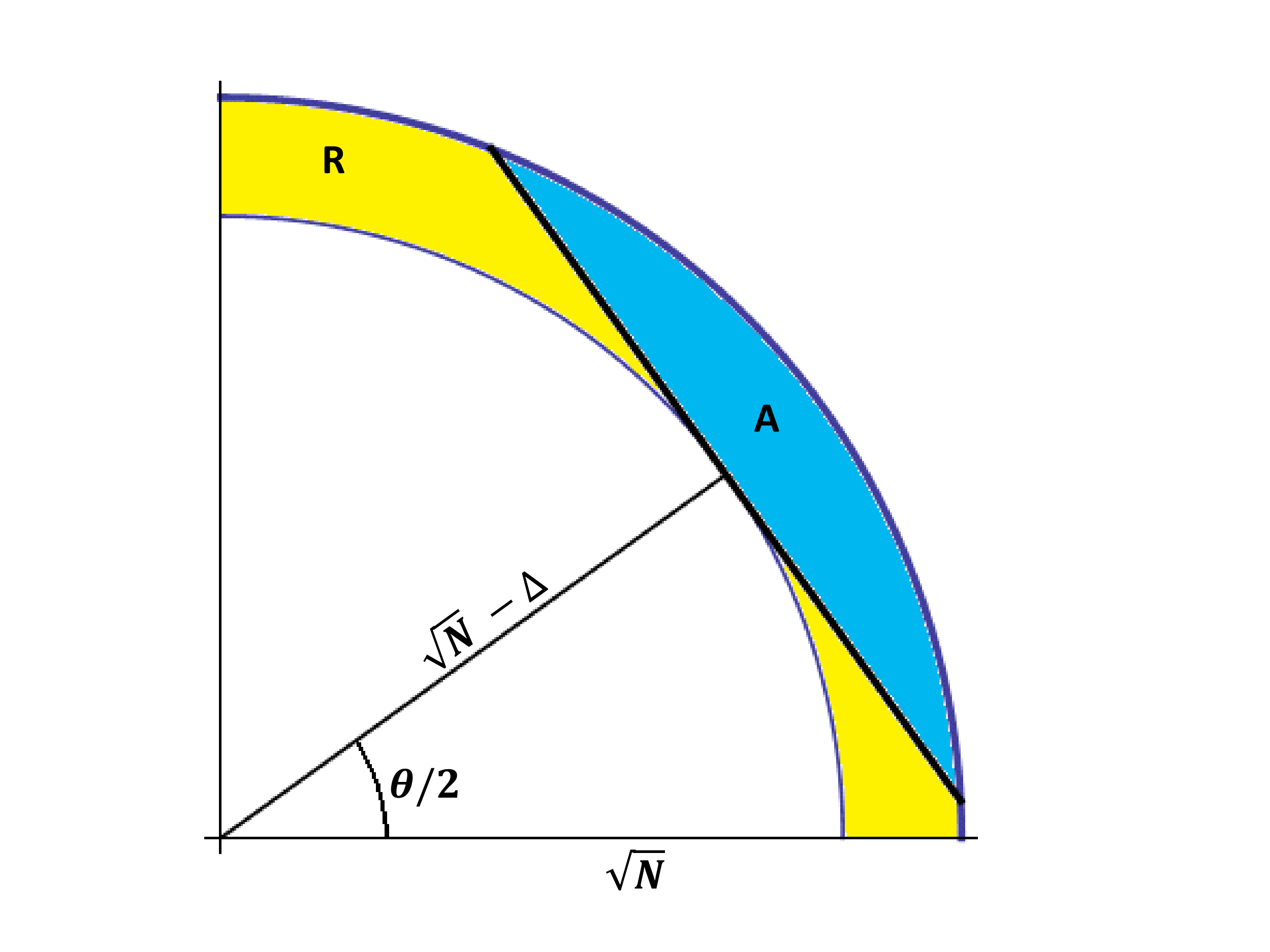}
  \caption[Conjecture 1 graphics]{The ring $R(N,\Delta)$ (yellow) and the segment $A(N,\Delta, P_{+})$ (blue), illustrated for values $N=625$ and $\Delta=4$.}
\label{fig:RingAndSector}
\end{figure}

We are concerned here with the segments of the standard integer grid that are contained in $R(N,\Delta)$ and $A(N,\Delta, P_{+})$, and their asymptotic behavior when $N \rightarrow \infty$.

We note that the Euclidean area $\mathcal{A}$ of $R(N,\Delta)$ is
\begin{equation*}
\mathcal{A}(R(N,\Delta)) = 2 \pi \, \Delta \sqrt{N} + O(\Delta^2)
\end{equation*}
and the Euclidean area of $A(N,\Delta, P_{+})$ is
\begin{equation*}
\mathcal{A}(A(N,\Delta, P_{+})) = 4/3 \, \Delta \sqrt{2 \Delta} N^{1/4} + O(\Delta^{5/2} N^{-1/4}).
\end{equation*}

Estimation of the number of integer grid points inside a flat contour is a known open problem with a rich history \cite{Guy}.
For our purposes, it suffices to know that the number of integer grid points
\begin{equation*}
\{ x,y \in \mathbb Z, (x,y) \in R(N,\Delta)\}
\end{equation*}
is asymptotically equal to $\Theta(\Delta \sqrt{N})$ and that the number of integer grid points
\begin{equation*}
\{ x,y \in \mathbb Z, (x,y) \in A(N,\Delta, P_{+})\}
\end{equation*}
is asymptotically equal to $\Theta(\Delta^{3/2} N^{1/4})$. These claims can be proven by elementary geometric means.

Finally, we assume that $N=p^{L}$ where $p$ is a fixed integer prime number with $p=1 \mod 4$ and $L$ is a large integer.

Consider the set
\begin{equation*}
s_4(N)=\{(x,y,z,w)\in {\mathbb Z}^4\ |\ x^2+y^2+z^2+w^2 = N\}
\end{equation*}
of all representations of $N$ as a sum of squares of four integers.
For $N=p^{L}$, the cardinality of the set is
\begin{equation*}
card(s_4(N)) = 8 (p^{L+1}-1)/(p-1) = O(N).
\end{equation*}
This is an immediate consequence of the formula expressing $|s_4(N)|$ as 8 times the sum of divisors of $N$ (see \cite{HW}).

The projection of $s_4(N)$ on the $(x,y)$ plane is contained in the circle of radius $\sqrt{N}$ and the projection of each point is an integer grid point in that circle.
The converse is not true:  for $N=p^{L}$ there are roughly $ 2 \, (p^{L+1}-1)/(p-1)$ projection points in the circle, while there are roughly $\pi \, p^L$ integer grid points (see \cite{Kraetzel}).

To determine the complexity of our first algorithm, we require a conjecture that states, informally, that the density of the $(x,y)$-projection points of $s_4(N)$ in the ring $R(N,\Delta)$ and the segment $A(N,\Delta, P_{+})$ is the same as the density of these projection points in the entire circle of radius $\sqrt{N}$.\footnote{To show that the algorithm requires expected polynomial time, a weaker form of Conjecture \ref{conj:1} may be considered;
the weaker claim is that the density of the projection points in the ring and the segment is at most polylogarithmically lower than their density in the circle.} The conjecture is motivated by the Corollary from Theorem 1 in Ref.~\onlinecite{Hooley}.
Although the conjecture is presented for a general prime $p = 1\mod 4$, our algorithms are developed for $p=5$, thus we require it to be true only for $p=5$.

\begin{conjecture}
\label{conj:1}
Consider $N=p^{L}$, where $p$ is a fixed integer prime number with $p=1 \mod 4$ and $L$ is a large even integer.
For a constant $\Delta > 1$, let the four-square decomposition set $s_4(N)$, the geometric ring $R(N,\Delta)$, and the circular segment $A(N,\Delta, P_{+})$ be defined as above.
Let $Pr_{x,y}(s_4(N))$ be the projection of the $s_4(N)$ onto its first two coordinates.
Then
\begin{eqnarray*}
&(1)&\ card\left(Pr_{x,y}(s_4(N)\right) \bigcap R(N,\Delta)) = \Theta\left(p^{L/2}/L\right)\\
&(2)&\ card\left(Pr_{x,y}(s_4(N)\right) \bigcap A(N,\Delta, P_{+})) = \Theta\left(p^{L/4}/L\right)
\end{eqnarray*}
\end{conjecture}


We conclude this subsection with number theory and experiments that support the conjecture.
Define the set
\begin{equation*}
sn(N,\Delta)=\{ a^2+b^2\ |\ a,b \in \mathbb Z, (a,b) \in R(N,\Delta)\}.
\end{equation*}
It is easy to see that
\begin{equation*}
sn(N,\Delta) \subset [p^L - 2 \Delta \, p^{L/2}, p^L].
\end{equation*}
If $2 \Delta \, p^{L/2} < p^L$, then the conditions of Thm 1 in Ref.~\onlinecite{Hooley} are satisfied and the Corollary implies that the cardinality of the set is
\begin{equation*}
card\left(sn(N,\Delta)\right) = \Theta\left(\frac{p^{L/2}}{\sqrt{\log(p^L)}}\right) = \Theta\left(\frac{p^{L/2}}{L^{1/2}}\right).
\end{equation*}
Thus, there are as many distinct circumferences in $R(N,\Delta)$ that contain integer grid points, implying that the number of integer grid points on any one of these circumferences is $\Omega(L^{1/2})$ on average.

We further note that the set
\begin{equation*}
v(L) = \{p^L - a^2 - b^2\ |\ a,b \in \mathbb Z, (a,b) \in R(N,\Delta) \}
\end{equation*}
has cardinality
\begin{equation*}
m = card\left(v(L)\right) = \Theta\left(\frac{p^{L/2}}{L^{1/2}}\right).
\end{equation*}
Values from $v(L)$ are contained in the interval $[0,2\, \Delta \, p^{L/2}]$.
The average density of integers in that segment that are representable as a sum of two squares of integers is $\Theta(\sqrt{log(N)}) = \Theta(\sqrt{L})$ \cite{Landau}.
Assuming that the density of such integers across the set $v(L)$ is the same, we infer from the assumption that there are $m/\sqrt{L}= \Theta(p^{L/2}/L)$ values in $v(L)$ that are so representable, and hence at least as many integer grid points $(a,b) \in R(N,\Delta)$ that are projections of some four square decomposition of $p^L$ ( i.e., such that there exist $c,d \in \mathbb Z$ with $p^L=a^2+b^2+c^2+d^2$).

To verify the statement (1) of Conjecture \ref{conj:1},
we ran extensive computer simulations for $p=5$ and $L=\{16,...,28\}$, and for $p=13$ and $L=\{12,...,18\}$, using Mathematica infinite precision integer arithmetic. 
and observed behavior consistent with the conjecture.
To motivate statement (2) of Conjecture \ref{conj:1}, we tested the polar angles of points in $Pr_{x,y}(s_4(N))$ for uniformity.
The simulation covered $N=5^{16}, ...,5^{28},13^{12},...,13^{18}$ and tested the null hypothesis that the distribution of the polar angles is uniform.
Based on Kolmogorov-Smirnov statistics, the null hypothesis could not be rejected at any meaningful level of significance.

\subsection{The Algorithm}


We now present the expected-polynomial time algorithm.
We begin by approximating an arbitrary $Z$-rotation with a $\langle\mathcal{C}+V \rangle$ circuit.


\begin{Problem}
\label{axialProblem}
Given a $Z$-rotation $G=R_Z(\theta)$ and a small enough \footnote{Although we do not have a closed form bound on how small $\epsilon$ should be, our algorithm works well in practice for $\epsilon < 2*5^{-4}=0.0032$.} target precision $\epsilon$, synthesize a $\langle\mathcal{C}+V \rangle$ circuit $c(G,\epsilon)$ such that
\be
dist(c(G,\epsilon), G) < \epsilon
\ee
and
\be
V_c(c(G,\epsilon))) \leq 4 \, \log_5(2/\epsilon)).
\ee
\end{Problem}

\begin{thm}
\label{thm:randalg}
There exists a randomized algorithm that solves Problem \ref{axialProblem} in expected time polynomial in $log(1/\epsilon)$.
\end{thm}

We first present geometry that relates the theorem to the Conjecture \ref{conj:1} with $p=5$.
Our goal is to select the target circuit depth value $L$ such that
\begin{equation}
\label{eq7}
\epsilon < 2 * 5^{-L/4}.
\end{equation}

Having found the smallest integer $L$ satisfying Eq (\ref{eq7}),
we then represent $G$ as $G = \cos\left(\frac{\theta}{2}\right) \, I + i \, \sin\left(\frac{\theta}{2}\right) Z$ and consider approximating it with
\begin{equation*}
U=\left( a\, I+b\, i\,X+c\, i\, Y+d\, i\, Z \right)5^{-L/2},
\end{equation*}
as suggested by Thm \ref{Proposition 1}.

Approximating $G$ to precision $\epsilon$ in the trace distance metric is equivalent to finding $U$ such that
\begin{equation*}
\left(a\cos\left(\frac{\theta}{2}\right)+d\sin\left(\frac{\theta}{2}\right)\right)5^{-L/2} > 1 - {\epsilon}^2.
\end{equation*}

For convenience we note that, without loss of generality, it suffices to prove the theorem for $-\pi/2 < \theta < \pi/2$ since we can always rotate the target gate to a position within this interval using $R_Z( \pm \pi/2)$ rotations from the Clifford group.
We also note that our selection of $L$ ensures that $5^{L/4} \, {\epsilon} \sim 2$.

Denote by $A_{\epsilon}(\theta)$ the segment of the unit disk where $(x\cos(\frac{\theta}{2})+y\sin(\frac{\theta}{2})) > 1 - {\epsilon}^2$.
Let $D(L)$ be an isotropic dilation of the plane with coefficient $5^{L/2}$.
Then the area of $D(L)[A_{\epsilon}(\theta)]$ is
\begin{equation*}
\mathcal{A}\left(D(L)[A_{\epsilon}(\theta)]\right) = 5^L \frac{4\sqrt{2}}{3} \epsilon^3  \sim 8\frac{4\sqrt{2}}{3} 5^{L/4}.
\end{equation*}

Define the angle $\phi = \sqrt{2} \epsilon (1-\epsilon^2/4)$ and the interval
\begin{equation*}
I_w(\epsilon,\theta) = \left(5^{L/2} \sin\left(\frac{\theta}{2}-\phi\right), 5^{L/2} \sin\left(\frac{\theta}{2}+\phi\right)\right)
\end{equation*}
with subinterval
$\left(5^{L/2} \sin\left(\frac{\theta}{2}- \epsilon\right), 5^{L/2} \sin\left(\frac{\theta}{2}+\epsilon\right)\right)$.

The length of the latter is approximately $2 * 5^{L/2} \cos\left(\frac{\theta}{2}\right) \epsilon \geq 2\sqrt{2} * 5^{L/4}$ and it contains approximately at least as many integer values.

Given any integer $a$ such that
\begin{equation*}
5^{L/2}\sin\left(\frac{\theta}{2}- \epsilon\right) < a < 5^{L/2}\sin\left(\frac{\theta}{2}+\epsilon\right)
\end{equation*}
we derive geometrically that the intersection of the horizontal line $w=a$ with $D(L)[A_{\epsilon}(\theta)]$ is a straight line segment that is longer than $5^{L/2}\frac{\epsilon^2}{2} \geq 2$ and that it contains at least two integer grid points.

We are now ready to prove the theorem.

\begin{proof}

Revisiting the notations of the previous subsection, we introduce the set of all representations of $5^L$ as a sum of squares of four integers
\begin{equation*}
s_4(5^L)=\{(x,y,z,w)\in {\mathbb Z}^4\ |\ x^2+y^2+z^2+w^2 = 5^L\}.
\end{equation*}

The key step in the algorithmic proof below is finding a point $(a,d)$ in the intersection of $Pr_{x,y}(s_4(5^L))$ and $D(L)[A_{\epsilon}(\theta)]$.

Once such a point is found we can use a Rabin-Shallit algorithm \cite{RS} to express $5^L -a^2-d^2$ as $b^2+c^2, b,c \in \mathbb Z$.
Then  $U=\left( a\, I+b\, i\,X+c\, i\, Y+d\, i\, Z \right)5^{-L/2}$ would be the desired approximation of $G$, that can be represented precisely as a $W$ circuit in at most $L$ quaternion division steps.

Consider a horizontal line $w=a$, where $a \in I_w(\epsilon,\theta)$.
By simple geometric calculation we find that the intersection of this line with the $D(L)[A_{\epsilon}(\theta)]$ segment is a line segment that is at most $5^{L/2} {\epsilon}^2 / \cos\left(\frac{\theta}{2}\right) \leq 5^{L/2} \, \sqrt{2} \, {\epsilon}^2$ long.

For our choice of $L$ this maximum length is approximately $4\sqrt{2}$ and thus the line segment contains at most 5 points with integer first coordinate.
On the other hand, we have shown earlier that for
\begin{equation*}
5^{L/2} \sin\left(\frac{\theta}{2}- \epsilon\right) < a < 5^{L/2} \sin\left(\frac{\theta}{2}+\epsilon\right)
\end{equation*}
the intersection of the $w=a$ line with $D(L)[A_{\epsilon}(\theta)]$ is a line segment that is longer than 2 and must contain at least 2 points with integer $z$ coordinate.
In other words, if $a \in I_w(\epsilon,\theta)$ is a randomly selected integer, then with probability at least $1/\sqrt{2}$ the intersection segment contains at least 2 integer grid points.

Algorithm \ref{alg:subroutine} gives the randomized approximation algorithm.
%
\begin{algorithm}[H]
\caption{Randomized Approximation}
\label{alg:subroutine}
\algsetup{indent=2em}
\begin{algorithmic}[1]
\REQUIRE{Accuracy $\epsilon$, angle $\theta$}
\STATE{completion $\leftarrow$ null}
\STATE{$Sw \, \leftarrow$ set of all integers in $I_w(\epsilon,\theta)$}
\WHILE{completion $==$ null \AND $Sw \ne \emptyset$}
\STATE{Randomly, pick an integer $a$ from $Sw$}
\STATE{$Sw \, \leftarrow \, Sw - \{a\}$}
\FORALL{integer $d$ such that $ (d,a) \in D(L)[A_{\epsilon}(\theta)]$}
\IF{exist $b,c \in \mathbb Z$\\
such that $5^L-a^2-d^2=b^2+c^2$}
\label{line}
\STATE completion $\leftarrow$ $(b,c)$
\STATE{Break;}
\ENDIF
\ENDFOR
\ENDWHILE
\IF {completion==null}
\RETURN{ null;}
\ENDIF
\STATE{$b \, \leftarrow$ first(completion)}
\STATE{$c \, \leftarrow$ last(completion)}
\RETURN{$U=\left( a\, I+b\, i\,X+c\, i\, Y+d\, i\, Z \right)5^{-L/2}$}
\end{algorithmic}
\end{algorithm}

In the worst case the algorithm terminates by exhausting the $\Theta(5^{L/4})$
candidate points in the $D(L)[A_{\epsilon}(\theta)]$ segment.
However, we note that this segment is that of Conjecture \ref{conj:1} for $p=5$.
Therefore the share of satisfactory candidates among all of the integer grid points $D(L)[A_{\epsilon}(\theta)]$ is $\Theta(1/L)$.
Thus the algorithm will terminate in $O(L)$ iterations on average.

Since the average overall number of iterations is moderate, the largest cost in the algorithm is line \ref{line}.
It has been shown by Rabin and Shallit \cite{RS} that the effective test for an integer $v$ to be a sum of squares of two integers has expected running cost of $O\left(\log^2(v) \, \log(\log(v))\right)$.
In our case $v \leq 8 * 5^{L/2}$ and we estimate the expected cost of the step as $O\left(L^2 \, \log(L)\right)$.
Therefore the overall expected cost of the algorithm is  $O\left(L^3 \, \log(L)\right)$ which translates into $O\left(\log(1/\epsilon)^3 \, \log(\log(1/\epsilon))\right)$.

\end{proof}

\subsection{Experimental Results}

We have implemented Algorithm \ref{alg:subroutine} from Thm \ref{thm:randalg} in Mathematica.
Our implementation has the following simplifications:
\begin{itemize}
\item Line \ref{line} has been redefined to return ${\tt PrimeQ}[5^L-a^2-d^2]$ for even $a$ and $d$ and to return \textrm{false} otherwise.\footnote{{\tt PrimeQ} is Mathematica primality test that does not require complete factorization of the integer being tested. Mathematica is a registered trademark of Wolfram Research, Inc.}

\item Given a desired $V$ count $V_c$, the algorithm terminates whenever a random candidate at distance less than $2 * 5^{-V_c/4}$ from the target is picked.
\end{itemize}
We
implicitly used the Rabin primality test since it is in general faster than complete integer factorization.
We ran our Mathematica solution over a set of 1000 random axial unitary rotations at 17 different circuit $V_c$ levels.
The test statistics are presented in Figure \ref{fig:Vpolynomial}.
The solid blue line represents the interpolated average precision achieved over the test set.
The sizes of the markers are proportional to the standard deviations of the precision at each level.
The dashed red line shows the theoretical precision bound of $2 * 5^{-V_c/4}$.
Note that the tight match between the theoretical estimate and experimental results is not very insightful since the algorithm has been designed to terminate as soon as the theoretical precision has been achieved.

\begin{figure}[tb]
  \centering
  \includegraphics[width=3.5in]{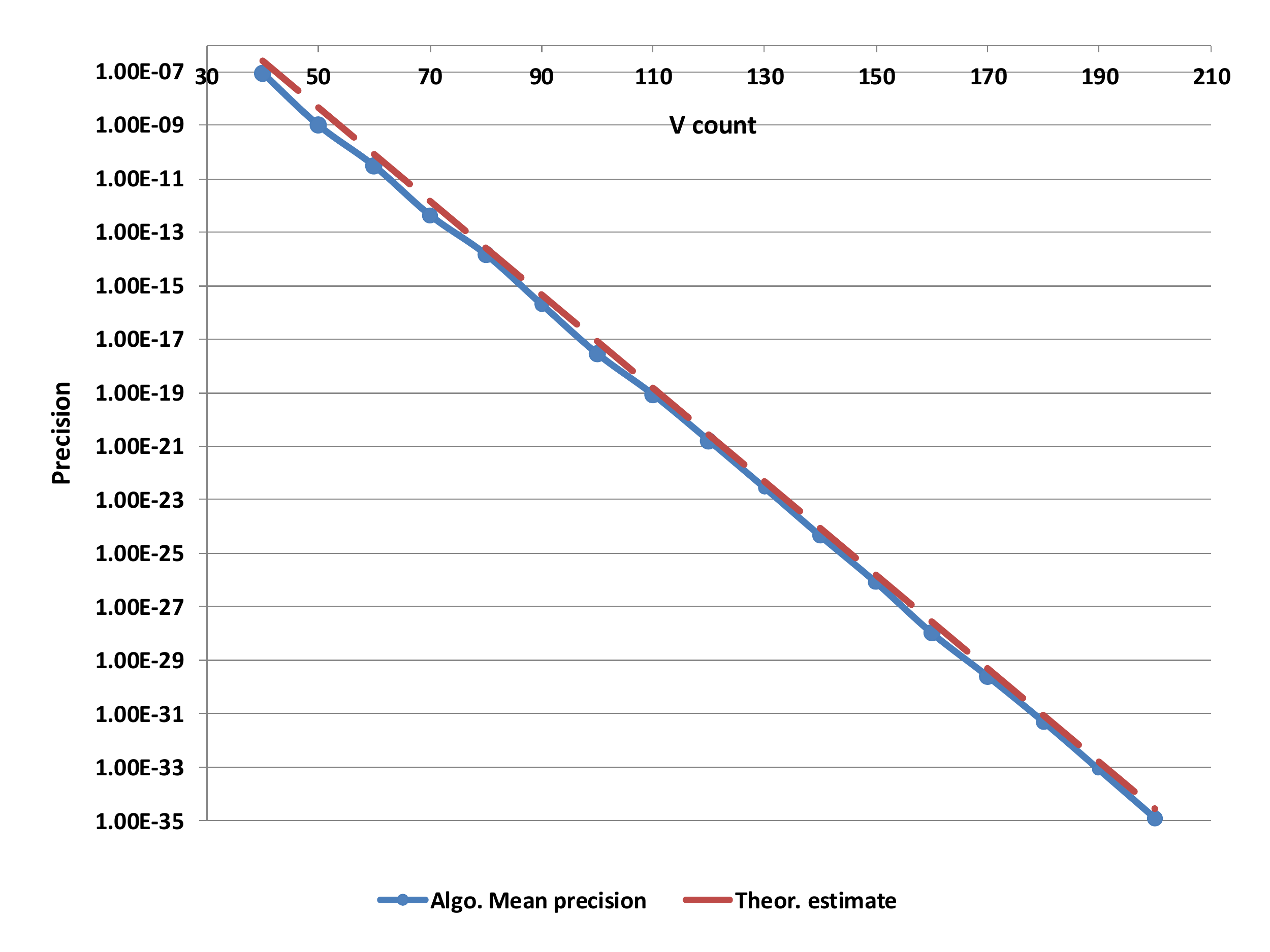}
  \caption[Clifford+V synthesis precision statistics]{$V$ count versus mean precision $\epsilon$ (measured in trace distance).
  Results are presented for $1000$ random axial rotations at 17 values of $V$ count.
  Solid blue line: interpolated average precision.
  Dashed red line: theoretical bound on precision, $2 * 5^{-V_c/4}$.
Marker sizes are proportional to the standard deviations of the precision at each $V$ count.
}
\label{fig:Vpolynomial}
\end{figure}

The algorithm can be used for approximate decomposition of any single-qubit unitary into a $\langle\mathcal{C}+V\rangle$ circuit
since any $G \in SU(2)$ can be decomposed exactly into three axial rotations, and the algorithm can be applied to each axial component.
The $V$ count in this case will scale as
\begin{equation}
\label{eq8}
V_c \leq 12\log_5(2/\epsilon).
\end{equation}

For the majority of unitary gates, we can significantly reduce the depth of the output circuit with a corresponding increase in compilation time.
The $V$ count estimate given in Eq (\ref{eq8}) reflects the tripling of the circuit depth due to decomposition of the target unitary into three axial rotations.
An alternative approach would be to perform a direct search in the four-dimensional integer grid; this will be the basis for our second algorithm described in Section \ref{sec:directalg}.

\subsection{A Possible Generalization}

A foundation for efficient circuit synthesis over the $V$ basis is the set of quaternions of norm $5^L$ and a body of number theory facts and conjectures related to that set.
Given an integer prime $p$ such that $p=1 \mod 4$, it is apparent that most of these facts and observations generalize to quaternions of norm $p^L$, which are generated by the primitive ones of norm $p$.
Modulo Lipschitz units there are $p+1$ such quaternions in the generator set.
These correspond to a basis of $p+1$ unitary operators that we denote $V(p)$.
Together with the Pauli gates a subset of $(p+1)/2$ of the $V(p)$ operators generate the generalization of the $W$ circuits.

However, in the case of $p=5$, it was sufficient to add only one $V$ operator in order to ensure the asymptotic uniformity of the grid of $\langle \mathcal{C}+\{V\} \rangle$ circuits.
For $p>5$, additional independent $V(p)$ operators are required.

For example, when $p=13$ the following gates are required, in addition to the Clifford group:
\begin{eqnarray*}
V_1(13) &=& (2 \,I + 3 \, i \, Z)/\sqrt{13}, \\
V_2(13) &=& (I + 2 \, i (X+Y+Z))/\sqrt{13}, \\
V_3(13) &=& (2 \, I + i\, X + 2 \, i (Y+Z))/\sqrt{13}, \\
V_4(13) &=& (2 \, I + i\, Y + 2 \, i (X+Z))/\sqrt{13},\\
V_5(13) &=& (2 \, I + i\, Z + 2 \, i (X+Y))/\sqrt{13}
\end{eqnarray*}

A generalization of Thm \ref{thm2} characterizes the gates representable exactly in the $\langle \mathcal{C}+\{V(p)\} \rangle$ basis as normalizations of Lipschitz quaternions of norm $p^L, L \in \mathbb Z$.
The exact synthesis of the corresponding circuit for a unitary of the form
\begin{equation*}
U=(a \, I + b \, i\, X + c \, i\, Y + d \, i\, Z)/p^{L/2}
\end{equation*}
amounts to a generalization of Algorithm \ref{alg:sub} and requires at most $(p+1)*L$ quaternion divisions.

Thm \ref{thm:randalg} also generalizes to $\langle \mathcal{C}+\{V(p)\} \rangle$ circuits,
to the extent that Conjecture \ref{conj:1} holds for the prime parameter $p$, and the circuit depth estimate from the theorem generalizes to an estimate of the form $L \leq 4\log_p(2/\epsilon)$.

We have chosen to focus on the $V(5)$ case for two reasons.
First, the basis requires only one non-Clifford gate for which we have a fault-tolerant implementation protocol
Second, we have so far only collected empirical data for $p=5$.

\section{Direct Search Approximation Algorithm}
\label{sec:directalg}


In this section, we present an algorithm based on optimized brute-force search for decomposing single-qubit unitaries into a circuit in the set $\langle \mathcal{P} + V \rangle$,
where $\mathcal{P}$ is the set of single-qubit Pauli gates and $V$ is one of the $V$ gates.
We first present relevant background.
%

\subsection{Vicinity of a Unitary in $PSU(2)$ as a Spherical Cap}

We begin by characterizing an $\epsilon$-neighborhood of a single qubit unitary as a ``spherical cap" in a 3-dimensional sphere $S^3$, i.e., as a portion of the sphere to one side of a certain 3-dimensional hyperplane in the 4-dimensional Euclidean space.

Consider the 4-dimensional Euclidean space with standard coordinates $\alpha, \beta, \gamma, \delta$.

Let
\begin{equation*}
S^3(R) = \{ (\alpha, \beta, \gamma, \delta)\ |\ \alpha^2+\beta^2+\gamma^2+\delta^2=R^2 \}
\end{equation*}
be  the 3-dimensional sphere of radius $R$ centered at the origin.
For any point on $S^3(R)$ we generate the unitary
\begin{equation*}
\nu(\alpha, \beta, \gamma, \delta) = (\alpha \, I + i \, \beta \, X + i \, \gamma \, Y + i \, \delta \, Z)/R \in SU(2).
\end{equation*}

The quantum gate group $PSU(2)$ is the central quotient of $SU(2)$ with the exact sequence
$1 \rightarrow {\mathbb Z}_2 \rightarrow SU(2) \rightarrow PSU(2) \rightarrow 1$, therefore $\nu$ defines a ${\mathbb Z}_2$ covering of $PSU(2)$
(which is the same factorization that is commonly used to glue an $S^3$ into 3-dimensional projective space).

Under this covering the $PSU(2)$ unitaries with nonzero trace are in one to one correspondence with the ``northern" hemisphere
\begin{equation*}
S^3_{+}(R) = \{ (\alpha, \beta, \gamma, \delta) \in S^3(R) | \alpha > 0 \}.
\end{equation*}

Thus given a gate $G,\ |tr(G)| > 0$, then a small enough $\epsilon$-vicinity of that gate
\begin{equation*}
c_{\epsilon}(G) = \{U \in PSU(2) |\ dist(U,G) < \epsilon\}
\end{equation*}
is unambiguously identified with a spherical cap in $S^3_{+}(R)$.

To clarify, consider $G=\nu(\alpha, \beta, \gamma, \delta)$  and define
\begin{eqnarray*}
C_{\epsilon}(G) = \{ (\alpha', \beta', \gamma', \delta') \in S^3_{+}(R)\ | \\
 \alpha\,\alpha'+ \beta\,\beta'+ \gamma \,\gamma'+ \delta\,\delta' > R(1-\epsilon^2)\}.
\end{eqnarray*}
Then $C_{\epsilon}(G)$ is a portion of $S^3(R)$ bounded by the hyperplane
\begin{equation*}
\alpha\,\alpha'+ \beta\,\beta'+ \gamma \,\gamma'+ \delta\,\delta' = R(1-\epsilon^2)
\end{equation*}
and $\nu(C_{\epsilon}(G)) = c_{\epsilon}(G)$.

We will focus further calculations on the $\epsilon$-neighborhoods that do not contain zero-trace gates and thus correspond to spherical caps completely contained in $S^3_{+}(R)$.
It is trivial to modify all of the equations to cases where an $\epsilon$-neighborhood intersects the zero-trace ``equator".

Given a $C_{\epsilon}(G)$ that is completely contained in $S^3_{+}(R)$, it is easy to derive, geometrically, that the metric volume $\mathcal{V}$ of $C_{\epsilon}(G)$ is
\begin{eqnarray*}
\mathcal{V}(C_{\epsilon}(G)) &= & 4\pi R^3 \int_{0}^{\cos^{-1}(\epsilon')} \sin^2(\eta)  d \eta \\
 &=& 2\pi R^3 \left(\cos^{-1}(\epsilon') - \frac{1}{2} \sin(2 \cos^{-1}(\epsilon'))\right),
\end{eqnarray*}
where $\epsilon' = 1 - \epsilon^2$.

Taking the Taylor series expansion of the latter at $\epsilon=0$, we find that
\begin{equation*}
\mathcal{V}(C_{\epsilon}(G)) = \frac{8 \pi \sqrt{2} \epsilon^3 R^3}{3} + O(\epsilon^5).
\end{equation*}
In the next sections we focus on precision targets $\epsilon$ for which the $C_{\epsilon}(G)$ neighborhoods have sufficient metric volume.

\subsection{A Bound for Uniform Precision}
We start by establishing that there exist unitary gates in $PSU(2)$ that cannot be approximated by $W$-circuits of $V_c \leq L$ to a precision better than $\epsilon_L = 5^{-L/4}/2$.
This is based on the following observation:
\begin{observation}
Let $w$ be a $W$-circuit different from the identity with $V_c(w) \leq L$, then it evaluates to $U(w)$ with
$|tr(U(w))| \leq 2(1-5^{-L/2})$ and $U(w)$ is at least $5^{-L/4}$ away from the identity.
\end{observation}

Indeed
\begin{equation*}
U(w) = \frac{a}{5^{L/2}} I + \frac{i \,(b\, X + c \, Y +  d \, Z)}{5^{L/2}},\ a,b,c,d \in \mathbb Z.
\end{equation*}
Since $U(w)$ is not the identity, $|a|$ cannot be greater than $5^{L/2}-1$.

Now, let $P \in \{I,X,Y,Z\}$ be a Pauli gate.
\begin{observation}
A circuit $w$ with $V_c(w) \leq L$ and distinct from $P$ evaluates to $U(w)$
with a distance at least $5^{-L/4}$ from $P$.
\end{observation}

Indeed, if $w$ is a $W$-circuit at a certain distance from $P$ then $w.P$ is a circuit with the same $V$ count at the same distance from the identity.

Thus, if
$\epsilon < \epsilon_L = 5^{-L/4}/2$ and $G \in PSU(2)$ is any unitary such that
\begin{equation*}
\epsilon < dist(G,P)< 2 \, \epsilon_L - \epsilon,
\end{equation*}
then there are no $W$-circuits of $V_c \leq L$ within distance $\epsilon$ from $G$ by the triangle inequality for $dist$:
\begin{equation*}
\forall w,\ dist(w,G) \geq dist(w,P) - dist(G,P) > \epsilon.
\end{equation*}

On the other hand, $dist(G,P)$ is also greater than $\epsilon$.
Therefore, the uniform precision guarantee cannot be better than $5^{-L/4}/2$ for $W$-circuits of $V_c \leq L$.
In other words, the uniform guarantee of optimal circuit depth cannot be better than $4 \log_5(1/\epsilon) - 4\log_5(2)$.

Revisiting the above discussions, we note that for $\epsilon < \epsilon_L = 5^{-L/4}/2$ there exist ``exclusion zones" of width $2(\epsilon_L - \epsilon)$ around each of the Pauli gates consisting of unitaries that cannot be approximated to precision $\epsilon$ by $W$-circuits with $V_c \leq L$.
Using the spherical cap volume formulae from the previous subsection, for $\epsilon$ significantly smaller than $\epsilon_L$, we estimate the combined volume of these exclusion zones, relative to the volume of the $S^3_{+}$ as $O(5^{-L/2}(5^{-L/4}-3\,\epsilon))$.

\subsection{A Working Conjecture}

Given the set of $W$-circuits with $V_c \leq L$, we will consider two key precision targets: $\epsilon_4(L) = 2 * 5^{-L/4}$ and
$\epsilon_3(L) = 5^{-L/3}$.

Consider the 3-dimensional hemisphere $S^3_{+}(5^{L/2})$.
As per the results from the previous subsection, for the metric volumes of the $\epsilon_4$- and $\epsilon_3$- neighborhoods we have
\begin{equation*}
\mathcal{V}(C_{\epsilon_4(L)}(G)) \sim 64 \pi\sqrt{2} 5^{3L/4} / 3
\end{equation*}
and
\begin{equation*}
\mathcal{V}(C_{\epsilon_3(L)}(G)) \sim 8 \pi\sqrt{2} 5^{L/2} / 3.
\end{equation*}

Since the volume of $S^3_{+}(5^{L/2})$ is equal to $\pi^2 5^{3L/2}$, the relative metric share that these neighborhoods occupy on the hemisphere are
\begin{equation*}
\frac{\mathcal{V}(C_{\epsilon_4(L)}(G))}{\mathcal{V}(S^3_{+}(5^{L/2}))} \sim \frac{64\sqrt{2} 5^{-3L/4}}{3 \pi}
\end{equation*}
and
\begin{equation*}
\frac{\mathcal{V}(C_{\epsilon_3(L)}(G))}{\mathcal{V}(S^3_{+}(5^{L/2}))} \sim \frac{8\sqrt{2}5^{-L}}{3 \pi},
\end{equation*}
respectively.

\begin{conjecture}
\label{conj:2}
(1) There exists a positive integer $L_4$ such that for any integer $L > L_4$ and any single-qubit gate $G$ there exists a $W$-circuit $w$ such that
\begin{equation*}
dist(G,w) \leq \epsilon_4(L).
\end{equation*}
(2) For large enough integer $L$  ($L > L_3$)  there exists an open subset
${\mathbb G}_3 \subset PSU(2)$ with metric volume $(1-o(1))\mathcal{V}(S^3_{+})$ (when $L \rightarrow \infty$)  such that for each
$G \in {\mathbb G}_3$ there exists a $W$-circuit $w$ with
\begin{equation*}
dist(G,w) \leq \epsilon_3(L).
\end{equation*}
\end{conjecture}

The common motivation for both clauses of this conjecture is that the number of distinct $W$-circuits scales as $5^{V_c}$.
More specifically, there are approximately $5 * 5^L$ distinct unitaries in $PSU(2)$ that are represented exactly by $W$-circuits with $V_c \leq L$.

This stems from the fact that $5^L$ has exactly $10(5^L - 2)$ distinct decompositions into a sum of four squares of integers, which can be easily derived from the Jacobi formula for the $r_4$ function:
\begin{equation*}
r_4(n) = 8 \sum_{SC(d)}d , SC(d)= (d | n) \& (d \mod 4 \ne 0)
\end{equation*}
(see chapters on the $r(n)$ function in Ref.~\onlinecite{HW}).
Geometrically, there are exactly $10(5^L -2)$ distinct integer grid points on $S^3(5^{L/2})$ and the set of such grid points is central-symmetrical with respect to the origin, so approximately half of these integer grid points lie on the $S^3_{+}(5^{L/2})$ piece of the hemisphere.

Further intuition in support of the conjectures is drawn from \cite{LPSI,LPSII}, which investigates the distribution density of the elements of the free group generated by $\langle V_1, V_2, V_3, V_1^{-1}, V_2^{-1}, V_3^{-1}\rangle$.

A stronger special case of Conjecture \ref{conj:2} postulates that for any
\begin{equation*}
G=\nu(\alpha, \beta, \gamma, \delta) = (\alpha \, I + i \, \beta \, X + i \, \gamma \, Y + i \, \delta \, Z),
\end{equation*}
where $\alpha^2+\beta^2+\gamma^2+\delta^2 = 1$,  there is an integer grid point on $S^3(5^{L/2})$ within distance $ \leq 2$ of $(\alpha, \beta, \gamma, \delta)*5^{L/2}$.
Although we do not claim that this stronger statement is true for all unitaries $G$, the perceived near-uniformness of the distribution of the integer lattice grid points over $S^3(5^{L/2})$ for large enough $L$ makes it plausible for most unitaries.

\subsection{Algorithm Outline}
Our algorithm to address Problem \ref{unitaryProblem} below employs optimized direct search.

\begin{Problem}
\label{unitaryProblem}
Given an arbitrary single-qubit unitary $G \in PSU(2)$ and a small enough target precision $\epsilon$, synthesize a $W$ circuit $c(G,\epsilon)$ such that
\be
dist(c(G,\epsilon), G) < \epsilon
\ee
and the $V$ count of the resulting circuit is
\be
\label{eq:general}
V_c \leq 3 \, \log_5(1/\epsilon)
\ee
for the majority of target unitaries and
\be
\label{eq:edge}
V_c \leq 4 \, \log_5(2/\epsilon)
\ee
in edge cases.
\end{Problem}

Let $L$ be the intended $V$ count of the desired approximation circuit.
Given a target single-qubit unitary gate represented as $G=\alpha I + \beta i X+\gamma i Y+\delta i Z$,
in order to find integers $(a,b,c,d)$ such that $a^2+b^2+c^2+d^2=5^L$ and
\begin{equation}
\label{eq2}
dist(G,(aI + biX+ciY+diZ)5^{-L/2}) < \epsilon,
\end{equation}
we split the ${\alpha,\beta,\gamma,\delta}$ coordinates into two-variable blocks.
Let us assume that the split is given by ${(\alpha,\delta),(\beta,\gamma)}$.
For the approximation inequality in Eq (\ref{eq2}) to hold it is sufficient that
\begin{equation}
\label{eq3}
{(b\, 5^{-L/2} - \beta)}^2 + {(c\, 5^{-L/2} - \gamma)}^2 < \epsilon^2
\end{equation}
and
\begin{equation}
\label{eq4}
{(a\, 5^{-L/2} - \alpha)}^2 + {(d\, 5^{-L/2} - \delta)}^2 < \epsilon^2.
\end{equation}

Our goal is to achieve $\epsilon=5^{-L/3}$.
It is easy to see that there are approximately $\pi 5^{L/3}$ integer pairs satisfying each of the conditions in Eq (\ref{eq3}) and Eq (\ref{eq4}) for that $\epsilon$.
We can now sweep over all of the $(b,c)$ integer pairs and build a hash table of all of the $5^L-b^2-c^2$ differences occurring in the first set.
Then we can sweep over all of the $(a,d)$ integer pairs from the second set, in search of one for which $a^2+d^2$ occurs in the hash table.

Using number-theoretical considerations (see, for example, \cite{RS}), one can reduce the number of candidates considered in this direct search by a factor of approximately $\frac{LR}{2 \sqrt{L\, \ln(5)}}$ (where $LR$ is the Landau-Ramanujan constant). Thus, for $L=34$ the reduction factor is approximately $0.05$.

For target unitaries that cannot be approximated to precision $5^{-L/3}$, the algorithm iteratively triples the precision goal (which has an effect of expanding the search space at each iteration) until the satisfactory candidate is found.

The outline of the algorithm is given in Algorithm \ref{alg:dsearch}.
\begin{algorithm}[H]
\caption{Direct Search Approximation}
\label{alg:dsearch}
\algsetup{indent=2em}
\begin{algorithmic}[1]
\REQUIRE{Accuracy $\epsilon$, Target gate $G=\alpha I + \beta i X+\gamma i Y+\delta i Z$ }
\STATE{$L \leftarrow \, \lfloor 3*\log_5(1/\epsilon) \rfloor$ }
\STATE{hash $\leftarrow \, Dictionary\langle Integer, (Integer*Integer) \rangle$}
\STATE{bound$\pm \leftarrow \, 5^L \, (\sqrt{\alpha^2 + \delta^2}\pm \epsilon)^2$ }
\FORALL {$b,c \in \mathbb{Z}$ satisfying Eq (\ref{eq3})}
\IF {bound$- \leq 5^L - b^2 - c^2 \leq $bound$+$ \AND $5^L-b^2-c^2$ is decomposable into two squares}
\STATE{Add $(5^L-b^2-c^2, (b,c)) \rightarrow $hash}
\ENDIF
\ENDFOR
\STATE{completion $\leftarrow$ fail}
\FORALL {integer pairs $(a,d)$ satisfying Eq (\ref{eq4}))}
\IF {hash contains key equal to $a^2+d^2$}
\STATE{completion $\leftarrow (a,b,c,d)$}
\STATE{Break;}
\ENDIF
\ENDFOR
\IF{completion $\ne$ fail }
\STATE{completion $\leftarrow$ completion$.(I,i \, X, i \, Y, i \, Z)5^{-L/2}$ }
\ENDIF
\RETURN{completion}
\end{algorithmic}
\end{algorithm}

\subsection{Experimental Results and Comparison}

\begin{figure}[tb]
  \centering
  \includegraphics[width=3.5in]{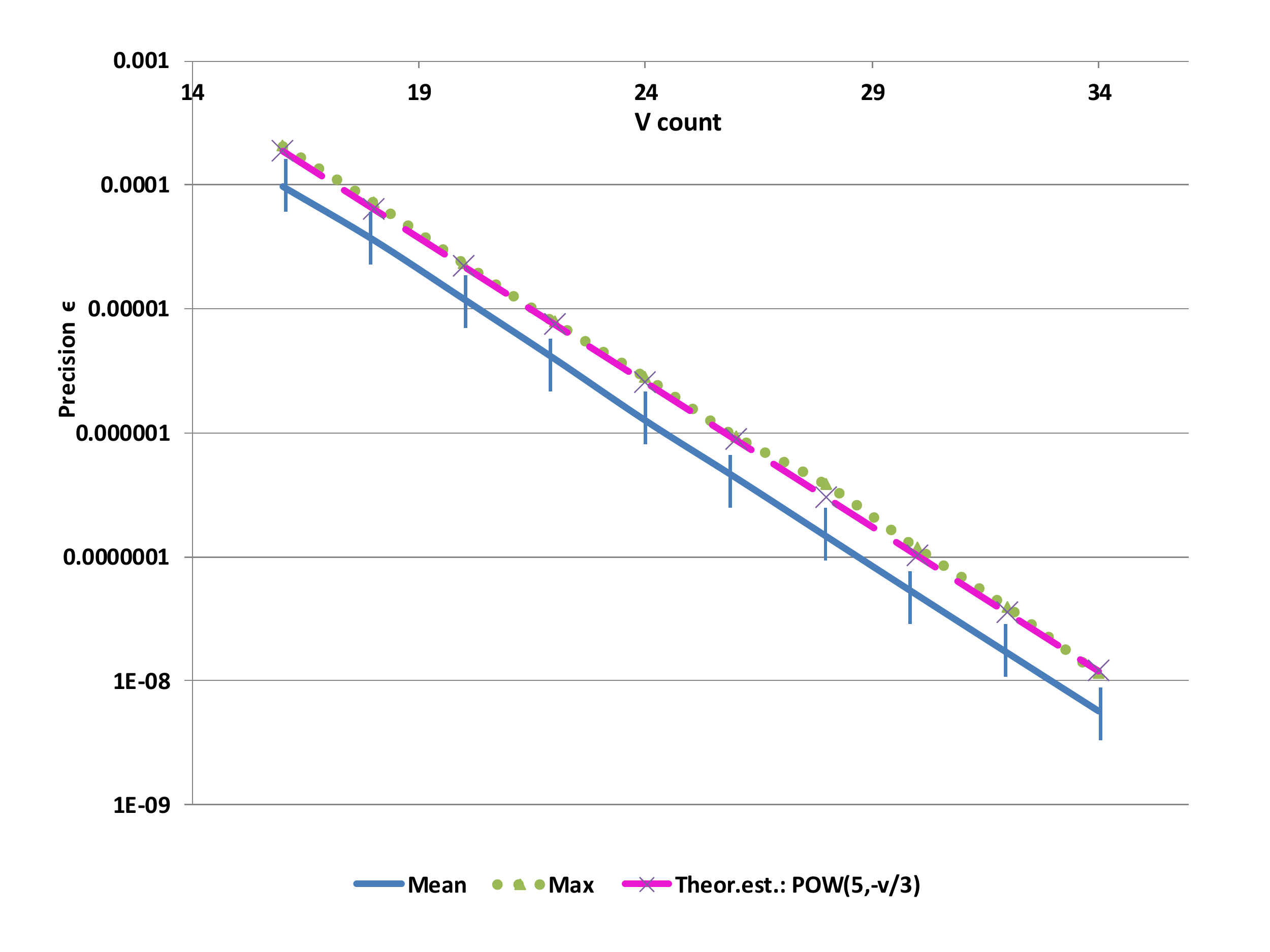}
  \caption[W-circuit synthesis precision statistics]{$V$ count versus mean precision $\epsilon$ (measured by trace distance) of the approximation of $1000$ random unitaries.
  The plot shows the $5^{-v/3}$ precision goal (dashed pink), the experimental average (solid blue), and the worst cases (dotted green).}
\label{fig:SKResults}
\end{figure}

The chart in Figure \ref{fig:SKResults} presents the results of evaluating our algorithm on a set of 1000 random unitaries.
The vertical axis plots precision $\epsilon$ on a logarithmic scale.
The horizontal axis plots the maximum $V$ count allowed in the approximating circuit.
The dashed pink curve represents the tight precision target of $5^{-V_c/3}$.
The solid blue curve represents the average approximation distance over the set of test unitaries; the error bars measure the standard deviation around the average.
The green dotted curve plots the worst cases.
For a small number of test unitaries, the algorithm could not find an approximating sequence with precision $5^{-v/3}$ or better for $V$ count $v$.


In practice, experimental evidence suggests that this algorithm works well for the majority of non-axial unitary rotations.
We have found that approximation circuits obtained by Algorithm \ref{alg:subroutine} are about 4 times deeper than the circuits produced by direct search using Algorithm \ref{alg:dsearch}.
This factor primarily arises because the non-axial rotation is first broken into three axial components and then a more liberal precision of $\epsilon_4(L)$ is pursued for each component.

Table \ref{tab:Algh1VsSearch} compares the $V$ count and precision values for Algorithms \ref{alg:subroutine}--\ref{alg:dsearch}, for precisions above $10^{-9}$.
The results support the rough factor of 4 reduction in $V$ count achieved by Algorithm \ref{alg:dsearch}.
The improvement is also apparent from the plot shown in Figure \ref{fig:RelWork} (green versus black curves).

\begin{table}[tb]
\begin{tabular}{|c|c|c|c|}
   \hline
   $\epsilon$ & RA, Median & DS, Median & DS, Worst\\
   \hline \hline
   $10^{-3}$ & 56.5 & 13 & 15\\
   $10^{-4}$ & 73.5 & 15.9 & 18\\
   $10^{-5}$ & 91 & 20.5 & 22\\
   $10^{-6}$ & 108 & 24.6 & 26\\
   $10^{-7}$ & 125 & 28.95 & 31\\
   $10^{-8}$ & 142.5 & 33.2 & 35\\
   $10^{-9}$ & 159.5 &  37.3 & 39\\
   \hline
 \end{tabular}
 \caption{$V$ counts for precision $\epsilon$ for two $V$ basis decomposition algorithms: Randomized Approximation (RA) (Algorithm \ref{alg:subroutine}) and Direct Search (DS) (Algorithm \ref{alg:dsearch}), for 1000 random non-axial rotations.
 Columns 2 and 3 list the median $V$ count; Column 4 lists the $V$ count for the worst case.}
   \label{tab:Algh1VsSearch}
 \end{table}

\section{Conclusions and Future Work}
In conclusion, we have proposed two novel algorithms for decomposing a single-qubit unitary into the $V$ basis, an efficiently universal basis that has some advantages over decomposing into the $\langle H,T\rangle$ basis.
Our algorithms produce efficient circuits that approximate a single-qubit unitary with high precision, and are computationally efficient in practice.  
Assuming a $V$ gate has the same cost as a $T$ gate, then our algorithms produce the shortest-depth approximation circuits known.

A key direction for future research is to determine a low-cost, exact implementation of a $V$ basis gate, which could include a native implementation on a given quantum computer architecture.
Discovery of such techniques would enable us to execute quantum circuits in $V$ basis at the quantum cost that is significantly lower than the cost of executing equivalent circuits  in the $\langle H,T\rangle$ basis.
It would also motivate further research into decomposition into other basis sets.

\begin{acknowledgments}
The Authors wish to thank Andreas Blass and Greg Martin for very useful discussions.
\end{acknowledgments}

\bibliography{BocharovSvore_PRL}

\begin{thebibliography}{17}%
\makeatletter
\providecommand \@ifxundefined [1]{%
 \@ifx{#1\undefined}
}%
\providecommand \@ifnum [1]{%
 \ifnum #1\expandafter \@firstoftwo
 \else \expandafter \@secondoftwo
 \fi
}%
\providecommand \@ifx [1]{%
 \ifx #1\expandafter \@firstoftwo
 \else \expandafter \@secondoftwo
 \fi
}%
\providecommand \natexlab [1]{#1}%
\providecommand \enquote  [1]{``#1''}%
\providecommand \bibnamefont  [1]{#1}%
\providecommand \bibfnamefont [1]{#1}%
\providecommand \citenamefont [1]{#1}%
\providecommand \href@noop [0]{\@secondoftwo}%
\providecommand \href [0]{\begingroup \@sanitize@url \@href}%
\providecommand \@href[1]{\@@startlink{#1}\@@href}%
\providecommand \@@href[1]{\endgroup#1\@@endlink}%
\providecommand \@sanitize@url [0]{\catcode `\\12\catcode `\$12\catcode
  `\&12\catcode `\#12\catcode `\^12\catcode `\_12\catcode `\%12\relax}%
\providecommand \@@startlink[1]{}%
\providecommand \@@endlink[0]{}%
\providecommand \url  [0]{\begingroup\@sanitize@url \@url }%
\providecommand \@url [1]{\endgroup\@href {#1}{\urlprefix }}%
\providecommand \urlprefix  [0]{URL }%
\providecommand \Eprint [0]{\href }%
\providecommand \doibase [0]{http://dx.doi.org/}%
\providecommand \selectlanguage [0]{\@gobble}%
\providecommand \bibinfo  [0]{\@secondoftwo}%
\providecommand \bibfield  [0]{\@secondoftwo}%
\providecommand \translation [1]{[#1]}%
\providecommand \BibitemOpen [0]{}%
\providecommand \bibitemStop [0]{}%
\providecommand \bibitemNoStop [0]{.\EOS\space}%
\providecommand \EOS [0]{\spacefactor3000\relax}%
\providecommand \BibitemShut  [1]{\csname bibitem#1\endcsname}%
\let\auto@bib@innerbib\@empty
\bibitem [{\citenamefont {Lubotsky}\ \emph {et~al.}(1986)\citenamefont
  {Lubotsky}, \citenamefont {Phillips},\ and\ \citenamefont {Sarnak}}]{LPSI}%
  \BibitemOpen
  \bibfield  {author} {\bibinfo {author} {\bibfnamefont {A.}~\bibnamefont
  {Lubotsky}}, \bibinfo {author} {\bibfnamefont {R.}~\bibnamefont {Phillips}},
  \ and\ \bibinfo {author} {\bibfnamefont {P.}~\bibnamefont {Sarnak}},\
  }\href@noop {} {\bibfield  {journal} {\bibinfo  {journal} {Comm. Pure and
  Appl. Math.}\ }\textbf {\bibinfo {volume} {34}},\ \bibinfo {pages} {149}
  (\bibinfo {year} {1986})}\BibitemShut {NoStop}%
\bibitem [{\citenamefont {Lubotsky}\ \emph {et~al.}(1987)\citenamefont
  {Lubotsky}, \citenamefont {Phillips},\ and\ \citenamefont {Sarnak}}]{LPSII}%
  \BibitemOpen
  \bibfield  {author} {\bibinfo {author} {\bibfnamefont {A.}~\bibnamefont
  {Lubotsky}}, \bibinfo {author} {\bibfnamefont {R.}~\bibnamefont {Phillips}},
  \ and\ \bibinfo {author} {\bibfnamefont {P.}~\bibnamefont {Sarnak}},\
  }\href@noop {} {\bibfield  {journal} {\bibinfo  {journal} {Comm. Pure and
  Appl. Math.}\ }\textbf {\bibinfo {volume} {40}},\ \bibinfo {pages} {401}
  (\bibinfo {year} {1987})}\BibitemShut {NoStop}%
\bibitem [{\citenamefont {Harrow}\ \emph {et~al.}(2002)\citenamefont {Harrow},
  \citenamefont {Recht},\ and\ \citenamefont {Chuang}}]{HRC}%
  \BibitemOpen
  \bibfield  {author} {\bibinfo {author} {\bibfnamefont {A.~W.}\ \bibnamefont
  {Harrow}}, \bibinfo {author} {\bibfnamefont {B.}~\bibnamefont {Recht}}, \
  and\ \bibinfo {author} {\bibfnamefont {I.~L.}\ \bibnamefont {Chuang}},\
  }\href {http://arxiv.org/abs/quant-ph/0111031} {\bibfield  {journal}
  {\bibinfo  {journal} {J. Math. Phys.}\ }\textbf {\bibinfo {volume} {43}}
  (\bibinfo {year} {2002})}\BibitemShut {NoStop}%
\bibitem [{\citenamefont {Selinger}(2012)}]{Selinger}%
  \BibitemOpen
  \bibfield  {author} {\bibinfo {author} {\bibfnamefont {P.}~\bibnamefont
  {Selinger}},\ }\href {http://arxiv.org/abs/1212.6253} {\enquote {\bibinfo
  {title} {Efficient clifford+{T} approximation of single-qubit operators},}\ }
  (\bibinfo {year} {2012}),\ \Eprint {http://arxiv.org/abs/1212.6253}
  {1212.6253} \BibitemShut {NoStop}%
\bibitem [{\citenamefont {Kliuchnikov}\ \emph
  {et~al.}(2012{\natexlab{a}})\citenamefont {Kliuchnikov}, \citenamefont
  {Maslov},\ and\ \citenamefont {Mosca}}]{KMM1231}%
  \BibitemOpen
  \bibfield  {author} {\bibinfo {author} {\bibfnamefont {V.}~\bibnamefont
  {Kliuchnikov}}, \bibinfo {author} {\bibfnamefont {D.}~\bibnamefont {Maslov}},
  \ and\ \bibinfo {author} {\bibfnamefont {M.}~\bibnamefont {Mosca}},\ }\href
  {http://arxiv.org/abs/1212.6964} {\enquote {\bibinfo {title} {Practical
  approximation of single-qubit unitaries by single-qubit quantum clifford and
  {T} circuits},}\ } (\bibinfo {year} {2012}{\natexlab{a}}),\ \Eprint
  {http://arxiv.org/abs/1212.6964} {1212.6964} \BibitemShut {NoStop}%
\bibitem [{\citenamefont {Duclos-Cianci}\ and\ \citenamefont
  {Svore}(2012)}]{DCS}%
  \BibitemOpen
  \bibfield  {author} {\bibinfo {author} {\bibfnamefont {G.}~\bibnamefont
  {Duclos-Cianci}}\ and\ \bibinfo {author} {\bibfnamefont {K.~M.}\ \bibnamefont
  {Svore}},\ }\href {http://arxiv.org/abs/1210.1980} {\enquote {\bibinfo
  {title} {A state distillation protocol to implement arbitrary single-qubit
  rotations},}\ } (\bibinfo {year} {2012}),\ \Eprint
  {http://arxiv.org/abs/1210.1980} {1210.1980} \BibitemShut {NoStop}%
\bibitem [{\citenamefont {Jones}\ \emph {et~al.}(2012)\citenamefont {Jones},
  \citenamefont {Whitfield}, \citenamefont {McMahon}, \citenamefont {Yung},
  \citenamefont {Meter}, \citenamefont {Aspuru-Guzik},\ and\ \citenamefont
  {Yamamoto}}]{CJetAl}%
  \BibitemOpen
  \bibfield  {author} {\bibinfo {author} {\bibfnamefont {N.~C.}\ \bibnamefont
  {Jones}}, \bibinfo {author} {\bibfnamefont {J.~D.}\ \bibnamefont
  {Whitfield}}, \bibinfo {author} {\bibfnamefont {P.~L.}\ \bibnamefont
  {McMahon}}, \bibinfo {author} {\bibfnamefont {M.-H.}\ \bibnamefont {Yung}},
  \bibinfo {author} {\bibfnamefont {R.~V.}\ \bibnamefont {Meter}}, \bibinfo
  {author} {\bibfnamefont {A.}~\bibnamefont {Aspuru-Guzik}}, \ and\ \bibinfo
  {author} {\bibfnamefont {Y.}~\bibnamefont {Yamamoto}},\ }\href
  {http://arxiv.org/abs/1204.0567} {\enquote {\bibinfo {title} {Simulating
  chemistry efficiently on fault-tolerant quantum computers},}\ } (\bibinfo
  {year} {2012}),\ \Eprint {http://arxiv.org/abs/1204.0567} {1204.0567}
  \BibitemShut {NoStop}%
\bibitem [{\citenamefont {Kliuchnikov}\ \emph
  {et~al.}(2012{\natexlab{b}})\citenamefont {Kliuchnikov}, \citenamefont
  {Maslov},\ and\ \citenamefont {Mosca}}]{KMM12}%
  \BibitemOpen
  \bibfield  {author} {\bibinfo {author} {\bibfnamefont {V.}~\bibnamefont
  {Kliuchnikov}}, \bibinfo {author} {\bibfnamefont {D.}~\bibnamefont {Maslov}},
  \ and\ \bibinfo {author} {\bibfnamefont {M.}~\bibnamefont {Mosca}},\ }\href
  {http://arxiv.org/abs/1206.5236} {\enquote {\bibinfo {title} {Fast and
  efficient exact synthesis of single qubit unitaries generated by clifford and
  t gates},}\ } (\bibinfo {year} {2012}{\natexlab{b}}),\ \Eprint
  {http://arxiv.org/abs/1206.5236} {1206.5236} \BibitemShut {NoStop}%
\bibitem [{\citenamefont {Kliuchnikov}\ \emph
  {et~al.}(2012{\natexlab{c}})\citenamefont {Kliuchnikov}, \citenamefont
  {Maslov},\ and\ \citenamefont {Mosca}}]{KMMb12}%
  \BibitemOpen
  \bibfield  {author} {\bibinfo {author} {\bibfnamefont {V.}~\bibnamefont
  {Kliuchnikov}}, \bibinfo {author} {\bibfnamefont {D.}~\bibnamefont {Maslov}},
  \ and\ \bibinfo {author} {\bibfnamefont {M.}~\bibnamefont {Mosca}},\ }\href
  {http://arxiv.org/abs/1212.0822} {\enquote {\bibinfo {title} {Asymptotically
  optimal approximation of single qubit unitaries by clifford and t circuits
  using a constant number of ancillary qubits},}\ } (\bibinfo {year}
  {2012}{\natexlab{c}}),\ \Eprint {http://arxiv.org/abs/1212.0822} {1212.0822}
  \BibitemShut {NoStop}%
\bibitem [{\citenamefont {Lipschitz}(1886)}]{Lip}%
  \BibitemOpen
  \bibfield  {author} {\bibinfo {author} {\bibfnamefont {R.}~\bibnamefont
  {Lipschitz}},\ }\href@noop {} {\emph {\bibinfo {title} {Untersuchungen uber
  die Summen von Quadraten}}}\ (\bibinfo  {publisher} {Bonn},\ \bibinfo {year}
  {1886})\BibitemShut {NoStop}%
\bibitem [{\citenamefont {Conway}\ and\ \citenamefont {Smith}(2003)}]{CS}%
  \BibitemOpen
  \bibfield  {author} {\bibinfo {author} {\bibfnamefont {J.}~\bibnamefont
  {Conway}}\ and\ \bibinfo {author} {\bibfnamefont {D.}~\bibnamefont {Smith}},\
  }\href@noop {} {\emph {\bibinfo {title} {On Quaternions And Octonions: Their
  Geometry, Arithmetic, And Symmetry}}}\ (\bibinfo  {publisher} {A K Peters,
  Ltd},\ \bibinfo {year} {2003})\BibitemShut {NoStop}%
\bibitem [{\citenamefont {Guy}(2004)}]{Guy}%
  \BibitemOpen
  \bibfield  {author} {\bibinfo {author} {\bibfnamefont {R.}~\bibnamefont
  {Guy}},\ }\href@noop {} {\emph {\bibinfo {title} {Unsolved problems in number
  theory}}},\ \bibinfo {edition} {3rd}\ ed.\ (\bibinfo  {publisher}
  {Springer},\ \bibinfo {year} {2004})\BibitemShut {NoStop}%
\bibitem [{\citenamefont {Hardy}\ and\ \citenamefont {Wright}(1979)}]{HW}%
  \BibitemOpen
  \bibfield  {author} {\bibinfo {author} {\bibfnamefont {G.}~\bibnamefont
  {Hardy}}\ and\ \bibinfo {author} {\bibfnamefont {E.}~\bibnamefont {Wright}},\
  }\href@noop {} {\emph {\bibinfo {title} {An Introduction to the Theory of
  Numbers}}}\ (\bibinfo  {publisher} {Oxford, Claredon Press},\ \bibinfo {year}
  {1979})\BibitemShut {NoStop}%
\bibitem [{\citenamefont {Kraetzel}(1988)}]{Kraetzel}%
  \BibitemOpen
  \bibfield  {author} {\bibinfo {author} {\bibfnamefont {E.}~\bibnamefont
  {Kraetzel}},\ }\href@noop {} {\emph {\bibinfo {title} {Lattice Points}}}\
  (\bibinfo  {publisher} {Kluwer, Dordrecht},\ \bibinfo {year}
  {1988})\BibitemShut {NoStop}%
\bibitem [{\citenamefont {Hooley}(1974)}]{Hooley}%
  \BibitemOpen
  \bibfield  {author} {\bibinfo {author} {\bibfnamefont {C.}~\bibnamefont
  {Hooley}},\ }\href@noop {} {\bibfield  {journal} {\bibinfo  {journal}
  {Journal fur die reine und angewandte Mathematik}\ }\textbf {\bibinfo
  {volume} {267}},\ \bibinfo {pages} {207} (\bibinfo {year}
  {1974})}\BibitemShut {NoStop}%
\bibitem [{\citenamefont {Landau}(1908)}]{Landau}%
  \BibitemOpen
  \bibfield  {author} {\bibinfo {author} {\bibfnamefont {E.}~\bibnamefont
  {Landau}},\ }\href@noop {} {\bibfield  {journal} {\bibinfo  {journal} {Arch.
  Math. Phys.}\ }\textbf {\bibinfo {volume} {13}} (\bibinfo {year}
  {1908})}\BibitemShut {NoStop}%
\bibitem [{\citenamefont {Rabin}\ and\ \citenamefont {Shallit}(1986)}]{RS}%
  \BibitemOpen
  \bibfield  {author} {\bibinfo {author} {\bibfnamefont {M.}~\bibnamefont
  {Rabin}}\ and\ \bibinfo {author} {\bibfnamefont {J.}~\bibnamefont
  {Shallit}},\ }\href@noop {} {\bibfield  {journal} {\bibinfo  {journal} {Comm.
  Pure and Appl. Math.}\ }\textbf {\bibinfo {volume} {39}},\ \bibinfo {pages}
  {239} (\bibinfo {year} {1986})}\BibitemShut {NoStop}%
\end{thebibliography}%

\appendix
\section{$V$ Gate Implementation}
\label{app:Vgate}
Any $V$ gate can be approximated using, for example, a $\langle H,T \rangle$ decomposition algorithm, but this is \emph{approximate} and requires a sequence length of 70 or more, depending on the desired precision.
In this appendix, we describe an \emph{exact} implementation of the $V$ gate using the protocol given in Ref.~\onlinecite{DCS}.
For additional details on the protocol, we refer the reader to Ref.~\onlinecite{DCS}.
This method can be used to implement any of the $V$ gates; here we show the implementation for $V_3$.

\subsection{Implementing $V_3$}

We implement the $V_3$ gate exactly, using a probabilistic circuit and a non-stabilizer resource state denoted as $\ket{H_2}$, where
\[
V_{3}=(I+2\, i\, Z)/\sqrt 5.
\]

In matrix form, this gate can be represented as:
\begin{equation*}
V_3 = \frac{1}{\sqrt 5}\left[
\begin{array}{cc}
 1+2i & 0 \\
 0 & 1-2i \\
\end{array}
\right]
= \frac{1+2i}{\sqrt 5}\left[
\begin{array}{cc}
 1 & 0 \\
 0 & \frac{-3}{5}-\frac{4}{5}i \\
\end{array}
\right]
\end{equation*}

Ignoring the global phase, we can solve for the angle of rotation $\theta$ about the $Z$ axis using the following identity:
\begin{eqnarray*}
e^{i\theta} &=& \cos\theta + i\sin\theta = \frac{-3}{5}-\frac{4}{5}i'\\
&\Rightarrow& \theta = \cos^{-1}(-\frac{3}{5}) \approx 4.06889.
\end{eqnarray*}

Consider the angle $\theta' = \cos^{-1}(\frac{3}{5}) \approx 0.927295$.
This angle is $\pi$ away from $\theta$: $\theta = \theta' + \pi$.
Thus, if we desire the rotation $Z(\theta)$, we can implement the gate sequence $Z(\theta) = Z(\theta')Z(\pi)$, where $Z(\pi)$ is the Pauli $Z$ gate.

Observe that
\begin{equation*}
\theta' = 2\theta_2+\frac{\pi}{4},
\end{equation*}
where $2\theta_2$ is the angle resulting from using the resource state $e^{-i\pi/8}HS^\dag\ket{H_2}$.
The $\frac{\pi}{4}$ part of the angle is a $T=Z(\pi/4)$ gate, thus $Z(\theta') = Z(2\theta_2)T$,
and $Z(\theta) = Z(2\theta_2)TZ$.

The circuit to obtain a rotation of $Z(2\theta_2)$ is given in Fig.~\ref{fig:H2Rot}.
The circuit results in the application of $\pm 2\theta_2$ to $\ket{\psi}$, each with equal probability.
If $m=0$, then $Z(2\theta_2)$ has been applied.
If $m=1$, we must apply $Z(4\theta_2)$.  Further details on the $m=1$ case are given in Section \ref{Sec:VCost}.

\begin{figure}[tb]
\[\Qcircuit @C=1em @R=1em {
\lstick{e^{-i\frac{\pi}{8}}HS^\dag\ket{H_2}} & \gate{X}  & \meter & \rstick{\ket{m}} \cw \\
\lstick{\ket{\psi}}      & \ctrl{-1} & \rstick{Z((-1)^m\, 2\theta_2)\ket{\psi}} \qw \\
}\]
\caption{Circuit to rotate by angle $\pm2\theta_2$ around the $Z$-axis.}
\label{fig:H2Rot}
\end{figure}
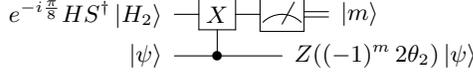

\subsection{Obtaining an $\ket{H_2}$ Resource State}

To implement $V_3$, we require a non-stabilizer state $\ket{H_2}$, which can be obtained using the ladder given in Ref.~\onlinecite{DCS}.
We begin by describing how to obtain the ladder state $\ket{H_2}$, and then describe how to implement $V_3$ using this resource state.

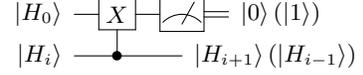
\begin{figure}[tb]
\[
\Qcircuit @C=1em @R=1em {
\lstick{\ket{H_0}} & \gate{X}  & \meter & \rstick{\ket{0} (\ket{1})} \cw \\
\lstick{\ket{H_i}} & \ctrl{-1} & \rstick{\ket{H_{i+1}} (\ket{H_{i-1}})} \qw \\
}\]
\caption{Two-qubit circuit used to obtain new $\ket{H_i}$ states from initial resource states $\ket{H_0}$. Upon measuring the 0 (1) outcome, the output state is $\ket{H_{i+1}}$ ($\ket{H_{i-1}}$).}\label{fig:2QbLadderCircs}
\end{figure}

The circuit of Fig.~\ref{fig:2QbLadderCircs} measures the parity of the two input qubits and decodes the resulting state into the second qubit.
Let the two inputs be magic states $\ket H$ and define $\theta_0=\frac\pi8$:
\begin{equation*}
\ket H = \ket{H_0} = \cos \theta_0 \ket 0 + \sin \theta_0 \ket 1.
\end{equation*}

Upon application of the controlled-{\tt NOT} gate $\Lambda(X)$,
\begin{eqnarray*}
  \ket{H_0}\ket{H_0} & \xrightarrow{\Lambda{(X)}} & \cos^2 \theta_0 \ket{00} + \sin^2 \theta_0 \ket{01}\\
  & & + \cos \theta_0 \sin \theta_0 (\ket{11} + \ket{10}).
\end{eqnarray*}

Upon measurement $m$ of the first qubit, we have
\begin{eqnarray*}
  \xrightarrow{m=0}&&  \frac{\cos^2\theta_0\ket0 + \sin^2\theta_0\ket1}{\cos^4\theta_0+\sin^4\theta_0}, \mbox{\ or\ } \\
  \xrightarrow{m=1}&&  \frac{1}{\sqrt2}(\ket0 + \ket1).
\end{eqnarray*}

We define $\theta_1$ such that
\begin{eqnarray*}
  \cos\theta_1\ket0 + \sin\theta_1\ket1& = &  \frac{\cos\theta_0\ket0 + \sin\theta_0\ket1}{\cos^4\theta_0+\sin^4\theta_0},
\end{eqnarray*}
from which we deduce
\begin{equation*}
\cot\theta_1=\cot^2\theta_0.
\end{equation*}

Thus we have
\begin{equation*}
\ket{H_1}=\cos\theta_1\ket0 + \sin\theta_1\ket1,
\end{equation*}
a non-stabilizer state obtained from $\ket H$ states, Clifford operations, and measurements.
If the measurement outcome is 1, then we obtain a stabilizer state
and discard the output (see Fig.~\ref{fig:2QbLadderCircs}).
The measurement outcomes occur with respective probabilities
$p_{m=0,0} = \cos^4 \theta_0 +\sin^4 \theta_0 =\frac{3}{4}$ and $p_{m=1,0} = 1-p_0 = \frac{1}{4}$.

Now consider the next step of the ladder.
We recurse on this protocol using the non-stabilizer states produced by the previous round of the protocol as input to the circuit in Fig.~\ref{fig:2QbLadderCircs}.
In this case, we need only go to state $\ket{H_2}$, which is defined as
\begin{eqnarray*}
\ket{H_2} = \cos \theta_2 \ket0 + \sin \theta_2 \ket1,
\end{eqnarray*}
where
\begin{equation*}
\cot\theta_2 = \cot^{3} \theta_0.
\end{equation*}

To obtain this state, we use as input the previously produced $\ket{H_1}$ state and a new $\ket{H_0}$ state:
\begin{eqnarray*}
  \ket{H_0}\ket{H_1} & \xrightarrow{\Lambda{(X)}} & \cos \theta_0 \cos \theta_1 \ket{00} + \sin \theta_0 \sin \theta_1 \ket{01}\\
  & & + \sin \theta_0 \cos \theta_1 \ket{10} + \cos \theta_0 \sin \theta_1 \ket{11}.
\end{eqnarray*}

Upon measurement of the first qubit, we have
\begin{eqnarray*}
  &\xrightarrow{m=0}&  (\cos \theta' \ket0 + \sin \theta' \ket1),\\
  &\xrightarrow{m=1}&  (\cos \theta'' \ket0 + \sin \theta'' \ket1), \mbox{\ where} \\
  \cot \theta'  &=& \cot \theta_1 \cot \theta_0 =\cot^{3} \theta_0 =\cot \theta_{2},\\
  \cot \theta''  &=& \cot \theta_1 \tan \theta_0 =\cot^{1} \theta_0 =\cot \theta_{0}.
\end{eqnarray*}

Thus, if we measure $m=0$, we obtain the state $\ket{H_{2}}$ and if we measure $m=1$, we obtain $\ket{H_{0}}$.
The probability of measuring 0 is given by
\begin{eqnarray*}
  p_{m=0,1} &=& \cos^2 \theta_1 \cos^2 \theta_0 +\sin^2 \theta_1 \sin^2 \theta_0.
\end{eqnarray*}

Note that $\frac{3}{4}\leq p_{m=0,i} <\cos^2 \frac\pi8=0.853\ldots$, so the probability of obtaining $\ket{H_2}$ is far higher than the probability of obtaining $\ket{H_0}$.


\subsection{Resource Cost}
\label{Sec:VCost}
What is the cost of obtaining a $\ket{H_2}$ state in terms of $\ket{H_0}$ resource states?
We simulated 10 million instances of the ladder to determine the average cost of obtaining $\ket{H_1}$ and $\ket{H_2}$.
Recall that the probabilities of moving ``up" the ladder are higher than moving ``down" the ladder.
For $\ket{H_1}$, the cost is on average $2.66$ $\ket{H_0}$ states, with a median cost of $2$.
For $\ket{H_2}$, the cost is on average $4.35$ $\ket{H_0}$ states, with a median cost of $3$.

What is the cost of implementing $Z(\theta')$?
Recall that our technique uses a probabilistic circuit with a success probability of $1/2$.
Thus, on average it will require two attempts for success.

If the circuit succeeds, the cost in $\ket{H_0}$ states is roughly $5.35$.
If the circuit fails, then we must correct the circuit by applying a $Z$ rotation of $2*2\theta_2$.
This requires preparing a resource state $Z(4\theta_2)$, which can be done using the circuit given in Fig.~\ref{fig:H2Rot} with $\ket{\psi} = e^{-i\pi/8}HS^\dag\ket{H_2}$.
On average, two attempts will be required to prepare the state, resulting in an average cost of 4 $\ket{H_2}$ states, or roughly $4*4.35 = 17.4$.
The prepared state is applied to the target qubit $\ket{\psi}$ using the same circuit in Fig.~\ref{fig:H2Rot}, except now the top input qubit is $\ket{Z(4\theta_2)}$.
The total cost if the circuit succeeds on this second attempt, after the first failure, is $1+4.35+17.4 = 22.75$.

As can be seen, each attempt that fails requires preparation of a more costly resource state for the next attempt.
The series of attempts is a negative binomial of parameter $p = \frac{1}{2}$ and the expected number of attempts to achieve success goes as $\sim \frac{1}{p} = 2$.
In general, at attempt $k$, a resource state to perform rotation by angle $2^k*2\theta_2$ is required.
The cost of preparing the resource state grows exponentially in $k$, and in the limit is infinite.
However, in practice, we will only make 1--3 attempts, and upon the final failure, apply a different approximation technique to the remaining rotation \footnote{We may in fact apply the backoff technique to the entire remaining sequence, that is, by determining the unitary from the remaining sequence and approximating it with the backoff technique.}, using methods of, for example, Refs.~\cite{Selinger,KMM1231}.
The optimal number of attempts to make before backing off to a different technique can be determined based on the required precision level (since the backoff method will only be approximate) and the chosen technique.

\end{document}